\documentclass{llncs}
\usepackage{url}
\usepackage{amsmath}
\usepackage{amssymb}
\usepackage{amsbsy}
\usepackage{stmaryrd}
\usepackage[ruled,linesnumbered,noline,noend]{algorithm2e}
\usepackage{tikz}
\usetikzlibrary{arrows,shapes}
\usepackage{pgf}
\usepgflibrary{plothandlers}

\newtheorem{dfn}{Definition}
\newtheorem{thm}{Theorem}
\newtheorem{lem}{Lemma}

\newcommand{\todo}[1]{{\scriptsize{{\textcolor{red}{[#1]}}}}}
\newcommand{\rem}[1]{{\scriptsize{{\textcolor{blue}{[#1]}}}}}
\renewcommand{\todo}[1]{}
\renewcommand{\rem}[1]{}

\newcommand{\false}{\ensuremath{\mathit{false}}}
\newcommand{\true}{\ensuremath{\mathit{true}}}

\newcommand{\var}[1]{{\tt #1}}

\newcommand{\prm}[1]{\ensuremath{\llbracket #1 \rrbracket}}
\newcommand{\ese}[1]{\ensuremath{\langle #1 \rangle}}
\newcommand{\bld}[1]{\ensuremath{\boldsymbol{#1}}}
\newcommand{\asgn}{\texttt{ := }}
\newcommand{\KwInSep}{ \\ \hspace{1cm} }
\newcommand{\KwOutSep}{ \\ \hspace{1.3cm} }
\newcommand{\lmark}[1]{\nlset{#1~~~~}}

\begin{document}

\title{Compact Symbolic Execution}
\author{Marek Trt\'{\i}k}
\institute{
    Faculty of Informatics, Masaryk University, Brno, Czech Republic\\
    \email{trtik@fi.muni.cz}
}
\maketitle

\begin{abstract}
    We present a generalisation of King's symbolic execution technique called compact symbolic execution. It is based on a concept of templates: a template is a declarative parametric description of such a program part, generating paths in symbolic execution tree with regularities in program states along them. Typical sources of these paths are program loops and recursive calls. Using the templates we fold the corresponding paths into single vertices and therefore considerably reduce size of the tree without loss of any information. There are even programs for which compact symbolic execution trees are finite even though the classic symbolic execution trees are infinite.
\end{abstract}

\section{Introduction} \label{sec:Introduction}

Classic symbolic execution as proposed by King in 1976~\cite{Kin76} systematically explores all real paths in an analysed program. There is typically huge (or even infinite) number of real paths even for very small and simple programs. Therefore, exploration of the real paths becomes a serious problem. We speak about the \emph{path explosion problem}.

Compact symbolic execution also explores all real program paths, but in a very compact manner. We analyse a given program \emph{before} we start its symbolic execution. We look for those parts of the program, which might produce real paths with some regularities in program states along them. Typically, program loops and recursion produces these regularities. We analyse the program parts independently from the remainder of the program. If the analysis of a part succeeds, then a result is a \emph{template}, i.e.~a declarative parametric description of the complete behaviour of the analysed part. Therefore, an output from the program analysis is a set of templates. Now we can execute the program symbolically with the templates. Until we reach some of the successfully analysed program parts, we proceed just like in classic symbolic execution. Let us now suppose we have just reached such a part. Having a template for the part, we do not need to symbolically execute interior of the part. We just instantiate the template into the end of the current path and then we \emph{jump} behind the part, where we continue with classic symbolic execution again.

Let us consider a symbolic execution reaching a loop. The execution may fork into a huge number of other symbolic executions during the execution of the loop. Each such execution has its own path in symbolic execution tree of classic symbolic execution. But having a template for the loop, we represent all these paths by a single one with the instantiated template. In other words, a single path explored by compact symbolic execution may represent a huge number of paths explored by classic symbolic execution. And that is the cause of the considerable space savings of compact symbolic execution. On the other hand, we will see that compact symbolic execution has higher requirements to performance of SMT solvers then classic one.

The worst case for compact symbolic execution is, when we fail to compute any template for a given program. Compact symbolic execution then reduces to classic one, and we gain no space savings.

\section{Overview} \label{sec:Overview}

In this section we give an intuition of compact symbolic execution. For simplicity of presentation we use the following definition of a program. Although our programs are simple they support typical imperative constructs and recursion.

\begin{dfn}[Program] \label{def:Program}
    A \emph{program} is a collection of functions and global variables. Each function has its own local variables. All program variables and functions have different names. Exactly one function is marked as starting one. Each function is represented as an oriented graph. Vertices in the graph identify program locations, while edges define transitions between them. We distinguish a single entry and exit location in each graph. There is no in-edge to entry location and there is no out-edge from the exit one. We label edges by actions to be taken when moving between connected locations. An action can be
    \begin{itemize}
        \item[(1)] An assignment of the form $<$variable$>${\tt :=}$<$expression$>$,

        \item[(2)] Call by value statements\\(a) $<$variable$>${\tt :=}$<$function-name$>${\tt (}$<$arg-list$>${\tt )}, or\\(b) $<$function-name$>${\tt (}$<$arg-list$>${\tt )}

        \item[(3)] A return value statement {\tt ret }$<$expression$>$,

        \item[(4)] {\tt skip} statement, which does nothing, or

        \item[(5)] A boolean expression over program variables.
    \end{itemize}
    If an edge $e = (u,v)$ is labelled by one of the actions (1)-(4), then out-degree of $u$ is $1$. Otherwise, label of $e$ is an action (5), out-degree of $u$ is $2$ and its out-edges are labelled by boolean expressions $\gamma$ and $\neg \gamma$. No action (2) can reference the starting function and no entry nor exit location is incident with an edge having an action (2). Each function $f$ is assigned a unique global variable $\textbf{ret}_f$ used for actions (2a) to save a return value being later assigned to the destination variable. And for simplicity we do not consider pointer arithmetic nor heap allocations. We prevent invalid operations in actions (like division by zero, etc.) by branchings into error locations. An error location is any location with a single out-edge heading back to that location and it is labelled with {\tt skip} action.
\end{dfn}

We can see an example of a program at Figure~\ref{fig:linSrch}~(a). The depicted function \texttt{linSrch} returns the least index \texttt{i} into the array \texttt{A} such that \texttt{A[i]==x}. If \texttt{x} is not in \texttt{A} at all, then it returns \texttt{-1}.

We first briefly describe classic symbolic execution as proposed by King~\cite{Kin76}. Instead of passing concrete data into parameters of the starting function, we pass symbols from a set $\{ \alpha_0, \alpha_1, \ldots \}$. Let us suppose we pass symbols $\alpha_0$ and $\alpha_1$ to variables \texttt{a} and \texttt{b} respectively. After executing an action \texttt{c:=2*a+b} the variable \texttt{c} will contain a \emph{symbolic expression} $2 \alpha_0 + \alpha_1$ as its value. \emph{Symbolic memory} is a function $\theta$ from program variables to a set of symbolic expressions. We further maintain a boolean symbolic expression $\varphi$ called \emph{path condition}. It represent a complete identifier of a particular program path taken during an execution. $\varphi$ is initially $\true$ and it can be updated at program branchings. Let $\theta$ be a symbolic memory having $\theta(\texttt{a}) = \alpha_0$, $\theta(\texttt{b}) = \alpha_1$ and $\theta(\texttt{c}) = 2 \alpha_0 + \alpha_1$ and let \texttt{c-a>2*b} and \texttt{c-a<=2*b} be actions of out-edges of an branching location. For the first action we proceed as follows. We evaluate the action in $\theta$. The result is a boolean symbolic expression $\alpha_0 + \alpha_1 > 2\alpha_1$. If $\varphi \rightarrow (\alpha_0 + \alpha_1 > 2\alpha_1)$ is satisfiable, we update $\varphi$ to $\varphi \wedge (\alpha_0 + \alpha_1 > 2\alpha_1)$ and we continue the execution by crossing the edge having the action. Then we proceed similarly for the second action. Note that if both implications are satisfiable, we fork the execution into two parallel and independent executions. Besides a symbolic memory and a path condition we commonly have a call stack $\Xi$ and we also need to identify a current program location $l$. Putting all the things together we get a \emph{program state} represented by a tuple $s = (\theta, \varphi, \Xi, l)$. Note that we understand a call stack record as pairs $(\sigma, l)$, where $l$ is a return location and $\sigma$ is a restriction of a symbolic memory to local variables. Further, we commonly describe the symbolic execution of a program by a tree structure called \emph{symbolic execution tree}. Vertices of the tree are related to program locations visited during the execution and edges reflect transitions between the locations. Each vertex of the tree is labelled by a related program state. But instead of labels $T$ and $F$ for branching edges (as proposed by King), we label them by evaluated actions of the branching edges. Figure~\ref{fig:linSrch}~(b) depicts a part of symbolic execution tree of the example program from Figure~\ref{fig:linSrch}~(a) (with omitted program states labelling the vertices). Please ignore grey regions in the tree for now. We assume that classic symbolic execution of the program started with an initial symbolic memory $\theta = \{ (\texttt{i},\alpha_0), (\texttt{n},\alpha_1), (\texttt{x},\alpha_2), (\texttt{A},\alpha_3) \}$.

We often use the following dot-notation to access elements of tuples. If $s = (\theta, \varphi, \Xi, l)$ is a program state, then $s.\theta$ denotes its symbolic memory, $s.\varphi$ denotes its path condition, $s.\Xi$ is its call stack and $s.l$ is a current program location. Further, if $u$ is a vertex of symbolic execution tree, then $u.s$ denotes program state labelling the vertex. And instead of $u.s.\theta$, $u.s.\varphi$, $u.s.\Xi$ and $u.s.l$ we simply write $u.\theta$, $u.\varphi$, $u.\Xi$ and $u.l$. Finally, if $\Xi$ is a call stack then we use dot-notation to access record at the top of the call stack. So, for example $\Xi.l$ denotes return location of record at the top of $\Xi$.

Symbols $\{ \alpha_0, \alpha_1, \ldots \}$ in classic symbolic execution represent input values to whole program. We generalise this concept to allow independent symbolic execution of \emph{parts} of an analysed program independently to the remainder. Each such a part uses the symbols $\{ \alpha_0, \alpha_1, \ldots \}$ relative to a chosen entry location to the part. Then using a \emph{composition} of program states (defined later) we can express any run of classic symbolic execution as a composition of program states resulting from analyses of the parts. Let $s = (\theta, \varphi, \Xi, l)$ be a program state resulting from a symbolic execution from a program location $l_0$ (e.g.~the entry location of the starting function), up to an entry location $l$ of an independently analysed program part. Let $s' = (\theta', \varphi', \Xi', l')$ be a program state resulting from the analysis of the part, i.e.~$s'$ represents a symbolic execution from the entry location $l$ to some exit location $l'$ from the part. Then $s \circ s' = (\theta \circ \theta', \varphi \wedge \theta\ese{\varphi'}, \Xi \circ (\theta \circ \Xi'), l')$ is composed program state representing symbolic execution from $l_0$ to $l'$ through the analysed part (entered in location $l$). We can see that composition of program states is implemented as composition of their individual components. We discuss very details of these operations in Section~\ref{sec:Definition}. Only note that composed path condition is $\varphi \wedge \theta\ese{\varphi'}$ rather then $\varphi \wedge \varphi'$. This is because $\varphi'$ may contain some symbols. But they are related to the entry location $l$ of the analysed part and not to the location $l_0$. Therefore, we have to compose $\varphi'$ with $\theta$ first to express $\varphi'$ in terms of symbols relative to location $l_0$. We do the similar effect of shifting symbols from location $l$ to $l_0$ in the compositions $\theta \circ \theta'$ and $\theta \circ \Xi'$.

\begin{figure}[!htb]
    \begin{tabular}{ccc}
        {\centering
\tikzstyle{start} = [regular polygon,regular polygon sides=3,
    regular polygon rotate=180,thick,draw,inner sep=1.6pt]
\tikzstyle{target} = [regular polygon,regular polygon sides=3,
    regular polygon rotate=0,thick,draw,inner sep=1pt]
\tikzstyle{loc} = [circle,thick,draw]
\tikzstyle{pre} = [<-,shorten <=1pt,>=stealth',semithick]
\tikzstyle{post} = [->,shorten <=1pt,>=stealth',semithick]
\footnotesize
\begin{tikzpicture}[node distance=1.0cm]
    \node [start] (a) {$a$};
    \node [loc] (b) [below of=a] {$b$}
        edge [pre] node [label=right:\texttt{i:=0}] {} (a);
    \node [loc] (c) [left of=b,yshift=-10mm] {$c$}
        edge [pre, bend left]
                node [label=above:\texttt{i<n~~~~}] {} (b);
    \node [loc] (f) [right of=b,yshift=-10mm] {$f$}
        edge [pre, bend right]
                node [label=above:\texttt{~~~~i>=n}] {} (b);
    \node [loc] (d) [right of=c,yshift=-10mm] {$d$}
        edge [pre, bend left]
            node [label=below:\texttt{A[i]!=x~~~~~~}] {} (c)
        edge [post, bend right=35]
                node [label=left:\texttt{++i}] {} (b);
    \node [loc] (e) [left of=c,yshift=-10mm] {$e$}
        edge [pre, bend left] node
            [label=above:\texttt{A[i]=x~~~~~}] {} (c);
    \node [target] (g) [below of=d,yshift=-5mm] {$g$}
        edge [pre, bend left=25]
                node [label=left:\texttt{ret i~~}] {} (e)
        edge [pre, bend right]
                node [label=left:\texttt{ret -1}] {} (f);
\end{tikzpicture} } & {\centering
\tikzstyle{start} = [regular polygon,regular polygon sides=3,
    regular polygon rotate=180,thick,draw,inner sep=0.5pt]
\tikzstyle{target} = [regular polygon,regular polygon sides=3,
    regular polygon rotate=0,thick,draw,inner sep=0.2pt]
\tikzstyle{loc} = [circle,thick,draw,inner sep=1.0pt]
\tikzstyle{pre} = [<-,shorten <=1pt,>=stealth',semithick]
\tikzstyle{post} = [->,shorten <=1pt,>=stealth',semithick]
\footnotesize
\begin{tikzpicture}[node distance=0.8cm]

    \pgfsetmovetofirstplotpoint

    \pgfplothandlerclosedcurve
    \pgfplotstreamstart
        \pgfplotstreampoint{\pgfpoint{0.2cm}{-0.5cm}}
        \pgfplotstreampoint{\pgfpoint{-1.1cm}{-1.3cm}}
        \pgfplotstreampoint{\pgfpoint{0.15cm}{-3cm}}
    \pgfplotstreamend
    \pgfsetfillcolor{lightgray!100}
    \pgfusepath{fill}
    \pgfplotstreamstart
        \pgfplotstreampoint{\pgfpoint{0.2cm}{-2.3cm}}
        \pgfplotstreampoint{\pgfpoint{-1.1cm}{-3.1cm}}
        \pgfplotstreampoint{\pgfpoint{0.15cm}{-4.8cm}}
    \pgfplotstreamend
    \pgfsetfillcolor{lightgray!100}
    \pgfusepath{fill}
    \pgfplothandlercurveto
    \pgfplotstreamstart
        \pgfplotstreampoint{\pgfpoint{-1.1cm}{-4.9cm}}
        \pgfplotstreampoint{\pgfpoint{0.15cm}{-4.0cm}}
        \pgfplotstreampoint{\pgfpoint{0.25cm}{-4.9cm}}
    \pgfplotstreamend
    \pgfsetfillcolor{lightgray!100}
    \pgfusepath{fill}

    \pgfplothandlerclosedcurve
    \pgfplotstreamstart
        \pgfplotstreampoint{\pgfpoint{0.2cm}{-0.5cm}}
        \pgfplotstreampoint{\pgfpoint{-1.1cm}{-1.3cm}}
        \pgfplotstreampoint{\pgfpoint{0.15cm}{-3cm}}
    \pgfplotstreamend
    \pgfusepath{stroke}
    \pgfplotstreamstart
        \pgfplotstreampoint{\pgfpoint{0.2cm}{-2.3cm}}
        \pgfplotstreampoint{\pgfpoint{-1.1cm}{-3.1cm}}
        \pgfplotstreampoint{\pgfpoint{0.15cm}{-4.8cm}}
    \pgfplotstreamend
    \pgfusepath{stroke}
    \pgfplothandlercurveto
    \pgfplotstreamstart
        \pgfplotstreampoint{\pgfpoint{-1.1cm}{-4.9cm}}
        \pgfplotstreampoint{\pgfpoint{0.15cm}{-4.0cm}}
        \pgfplotstreampoint{\pgfpoint{0.25cm}{-4.9cm}}
    \pgfplotstreamend
    \pgfusepath{stroke}

%

    \node [loc,inner sep=1.5pt] (a) {$a$};

    \node [loc] (b) [below of=a] {$b$}
        edge [pre] node {} (a);
    \node [loc] (c) [left of=b,yshift=-5mm] {$c$}
        edge [pre] node [label=above:{\scriptsize $0<\alpha_1~~~~~~$}] {} (b);
    \node [loc,inner sep=0.5pt] (f) [right of=b,xshift=2mm,yshift=-5mm] {$f$}
        edge [pre] node [label=above:{\scriptsize $~~~~0\geq\alpha_1$}] {} (b);
    \node [loc] (g) [below of=f] {$g$}
        edge [pre] node {} (f);
    \node [loc] (d) [right of=c,yshift=-5mm] {$d$}
        edge [pre] node [label=above:{\scriptsize
            $~~~~~~~~~~\alpha_3(0)\neq\alpha_2$}] {} (c);
    \node [loc] (e) [left of=c,yshift=-5mm] {$e$}
        edge [pre] node [label=above:{\scriptsize
            $\alpha_3(0)=\alpha_2~~~~~~~~~~$}] {} (c);
    \node [loc] (h) [below of=e] {$g$}
        edge [pre] node {} (e);

    \node [loc] (b1) [below of=d] {$b$}
        edge [pre] node {} (d);
    \node [loc] (c1) [left of=b1,yshift=-5mm] {$c$}
        edge [pre] node [label=above:{\scriptsize $1<\alpha_1~~~~~~$}] {} (b1);
    \node [loc,inner sep=0.5pt] (f1)[right of=b1,xshift=2mm,yshift=-5mm]{$f$}
        edge [pre] node [label=above:{\scriptsize $~~~~1\geq\alpha_1$}] {} (b1);
    \node [loc] (g1) [below of=f1] {$g$}
        edge [pre] node {} (f1);
    \node [loc] (d1) [right of=c1,yshift=-5mm] {$d$}
        edge [pre] node [label=above:{\scriptsize
            $~~~~~~~~~~\alpha_3(1)\neq\alpha_2$}] {} (c1);
    \node [loc] (e1) [left of=c1,yshift=-5mm] {$e$}
        edge [pre] node [label=above:{\scriptsize
            $\alpha_3(1)=\alpha_2~~~~~~~~~~$}] {} (c1);
    \node [loc] (h1) [below of=e1] {$g$}
        edge [pre] node {} (e1);

    \node [loc] (b2) [below of=d1] {$b$}
        edge [pre] node {} (d1);

    \node [] (b3) [below of=b2] {}
        edge [dotted,thick] node {} (b2);
\end{tikzpicture} } & {\centering
\tikzstyle{start} = [regular polygon,regular polygon sides=3,
    regular polygon rotate=180,thick,draw,inner sep=0.5pt]
\tikzstyle{target} = [regular polygon,regular polygon sides=3,
    regular polygon rotate=0,thick,draw,inner sep=0.2pt]
\tikzstyle{loc} = [circle,thick,draw,inner sep=1.0pt]
\tikzstyle{cloc} = [rectangle,thick,draw,inner sep=3pt]
\tikzstyle{pre} = [<-,shorten <=1pt,>=stealth',semithick]
\tikzstyle{post} = [->,shorten <=1pt,>=stealth',semithick]
\footnotesize
\begin{tikzpicture}[node distance=0.8cm]
%

     \node [loc,inner sep=1.5pt] (7)
                {$a$};
     \node [cloc] (8) [below of=7] {$b$}
         edge [pre] node {} (7);
     \node [loc,inner sep=1.5pt] (9) [below of=8, left of=8,yshift=-0.5cm] {$e$}
         edge [pre] node [label=left:{
                $\gamma_e$}] {} (8);
     \node [loc,inner sep=0.3pt] (10) [below of=8, right of=8,yshift=-0.5cm] {$f$}
         edge [pre] node [label=right:{
                $\gamma_f$}] {} (8);
     \node [loc,inner sep=1.0pt] (11) [below of=9] {$g$}
         edge [pre] node {} (9);
     \node [loc,inner sep=1.0pt] (12) [below of=10] {$g$}
         edge [pre] node {} (10);

     \node [rectangle] (b) [below of=7,xshift=3.5mm] {$s$};
     \node [rectangle] (e) [above of=11,xshift=4mm] {$s_e'$};
     \node [rectangle] (f) [above of=12,xshift=-3.8mm,yshift=-0.4mm] {$s_f'$};
\end{tikzpicture} }
        \\
        (a) & (b) & (c)
    \end{tabular}
    \caption{(a) A program with a function \texttt{linSrch(A,n,x)}. (b) Symbolic execution tree of function \texttt{linSrch}. (c) Compact symbolic execution tree of function \texttt{linSrch}.}
    \label{fig:linSrch}
\end{figure}

In symbolic execution tree at Figure~\ref{fig:linSrch}~(b) there is a single path highlighted by a sequence of grey regions. Vertices in each region are related to the same sequence of program locations: $b,c,d,b$. Moreover, we enter the path at vertex referencing location $b$ and we can leave the path either by stepping into a vertex referencing location $e$ or into a vertex referencing location $f$. Let us denote the entry vertex into the path as $b_0$ and the exit vertices from the path referencing locations $e$ and $f$ as $e_0, e_1, \ldots$ and $f_0, f_1, \ldots$ respectively being indexed from the top down. Our goal is to completely eliminate the path in grey from the tree, while still representing all real program paths. One way to do so is to represent whole the path by a single vertex, $b$ say, with two direct successors. The first successor, $e$ say, represents all the exit vertices $e_i$ from the path and the second, $f$ say, representing all the exits $f_i$. Note that names of the vertices $b$, $e$ and $f$ also represent program locations they reference. We label the vertex $b$ by the program state labelling $b_0$. But the question is what program states we should assign to the vertices $e$ and $f$. Note that two different vertices $e_i$ and $e_j$ may be labelled by different program states. So, for the vertex $e$ we need to introduce a program state $e.s\prm{\kappa}$, parametrised by a \emph{parameter} $\kappa$, such that each program state $e_i.s$ can be equivalently expressed by $e.s\prm{\kappa}$, when $\kappa$ is substituted by some number $\nu$. Of course, for different states $e_i.s$ and $e_j.s$ there are different numbers, say $\nu_i$ and $\nu_j$, for parameter substitution. We similarly need a parametrised program state $f.s\prm{\kappa}$ for the vertex $f$. We compute the states $e.s\prm{\kappa}$ and $f.s\prm{\kappa}$ before we start symbolic execution of the program from the Figure~\ref{fig:linSrch}~(a) by analysing the following its part. The part consists of all the locations $b,c,d,e,f$ discussed above and of all the edges between them. Note that the sequence $b,c,d,b$ of locations forms a cyclic path inside the analysed part. This cycle is actually the source of the path in grey regions. Nevertheless, we want to describe program states at \emph{exits} form the part. The exits from the part are target vertices of those edges of the part, which do not belong to the cycle. Therefore, locations $e$ and $f$ are the exits from the part. We also identify the location $b$ as entry location into the part, since we can enter the part by stepping into location $b$. The part is completely defined now. We analyse it independently from the remainder of the program. It mainly means that if we use some symbols $\alpha_i$ in the analysis, then they are related to the entry location $b$ of the part and not to the entry location of the whole program. At this point we are more concerned about formulation of a result from the analysis and its usage then the analysis itself. Therefore, we postpone its description to Section~\ref{sec:Detection}. We assume here that key properties $e.s\prm{\kappa}$ and $f.s\prm{\kappa}$ from the analysis are already computed, so we may formulate an output from the analysis of the part as the following \emph{template}
\[t = (b, 2, \{ (\theta_e, \varphi_e, [], e)\prm{\kappa}, (\theta_f, \varphi_f, [], f)\prm{\kappa} \}),\]
where $b$ is the entry location to the analysed part, the number $2$ identifies number of following parametrised program states and the remaining two tuples are the parametrised program states $e.s\prm{\kappa}$ and $f.s\prm{\kappa}$ respectively. Note that $[]$ identifies empty call stack. The template contains all the information we need to build compact symbolic execution tree, where the path in grey is folded as described above.

Let us symbolically execute the program at Figure~\ref{fig:linSrch}~(a) with the template $t$. We construct a compact symbolic execution tree during the execution. The tree is depicted at Figure~\ref{fig:linSrch}~(c). We apply classic symbolic execution, until we reach the entry location $t.b$. Let $b$ be the vertex in the tree, when we reach the location $t.b$ and let $s$ be the program state $b.s$. We now instantiate the template. Since we have exactly two program states in $t$, we create exactly two successor vertices $e$ and $f$ of the vertex $b$ in the tree. The vertices $e$ and $f$ references locations $t.e$ and $t.f$ respectively and they are further labelled by program states $s~\circ~(t.\theta_e,~t.\varphi_e,~[],~t.e)\prm{\kappa}$ and $s~\circ~(t.\theta_f,~t.\varphi_f,~[],~t.f)\prm{\kappa}$ respectively. We finish the instantiation of $t$ by creating edges $(b,e)$ and $(b,f)$ labelled by symbolic expressions $s.\theta\ese{t.\varphi_e\prm{\kappa}}$ and $s.\theta\ese{t.\varphi_f\prm{\kappa}}$ respectively. The situation is also depicted at Figure~\ref{fig:linSrch}~(c). Then we continue from both vertices $e$ and $f$ independently using classic symbolic execution again. These both executions reaches function exit location $g$ in one step and compact symbolic execution terminates.

\begin{figure}[!htb]
    \begin{tabular}{cc}
        {\centering
\tikzstyle{start} = [regular polygon,regular polygon sides=3,
    regular polygon rotate=180,thick,draw,inner sep=2.0pt]
\tikzstyle{target} = [regular polygon,regular polygon sides=3,
    regular polygon rotate=0,thick,draw,inner sep=1pt]
\tikzstyle{loc} = [circle,thick,draw]
\tikzstyle{pre} = [<-,shorten <=1pt,>=stealth',semithick]
\tikzstyle{post} = [->,shorten <=1pt,>=stealth',semithick]
\footnotesize
\begin{tikzpicture}[node distance=1.1cm]
    \node [start] (a) {$a$};
    \node [loc] (b) [below of=a] {$b$}
        edge [pre] node [label=right:\texttt{k:=0}] {} (a);
    \node [loc] (c) [below of=b] {$c$}
        edge [pre] node [label=right:\texttt{i:=0}] {} (b);
    \node [loc] (d) [left of=c,yshift=-10mm] {$d$}
        edge [pre, bend left] node [label=above:\texttt{i<n~~~~}] {} (c);
    \node [loc] (e) [left of=d,yshift=-10mm] {$e$}
        edge [pre, bend left] node [label=left:\texttt{A[i]=x}] {} (d);
    \node [loc] (f) [right of=d,yshift=-10mm] {$f$}
        edge [pre, bend right]
            node [label=below:\texttt{A[i]!=x~~~~~~~}] {} (d)
        edge [pre, bend left] node [label=above:\texttt{++k}] {} (e)
        edge [post, bend right=35] node [label=left:\texttt{++i}] {} (c);
    \node [loc] (g) [right of=c,yshift=-10mm] {$g$}
        edge [pre, bend right] node [label=above:\texttt{~~~~i>=n}] {} (c);
    \node [target] (h) [below of=g] {$h$}
        edge [pre] node [label=right:\texttt{ret k}] {} (g);
\end{tikzpicture}} & {\centering
\tikzstyle{start} = [regular polygon,regular polygon sides=3,
    regular polygon rotate=180,thick,draw,inner sep=0.5pt]
\tikzstyle{target} = [regular polygon,regular polygon sides=3,
    regular polygon rotate=0,thick,draw,inner sep=-0.4pt]
\tikzstyle{loc} = [circle,thick,draw,inner sep=1pt]
\tikzstyle{pre} = [<-,shorten <=1pt,>=stealth',semithick]
\tikzstyle{post} = [->,shorten <=1pt,>=stealth',semithick]
\footnotesize
\begin{tikzpicture}[node distance=0.7cm]
    \node [loc,inner sep=1.5pt] (1) {$a$};
    \node [loc] (2) [below of=1] {$b$}
        edge [pre] node {} (1);

    \node [rectangle,thick,draw] (3) [below of=2] {$c$}
        edge [pre] node {} (2);
    \node [loc] (4) [left of=3] {$g$}
        edge [pre] node [label=above:{$\gamma_{g}^{1}$}] {} (3);
    \node [loc] (5) [below of=4] {$h$}
        edge [pre] node {} (4);

    \node [loc,inner sep=0.4pt] (6) [below of=3,yshift=-1mm] {$f$}
        edge [pre] node [label=right:{$\gamma_{f}^{1}$}] {} (3);
    \node [rectangle,thick,draw] (7) [below of=6] {$c$}
        edge [pre] node {} (6);
    \node [loc] (8) [left of=7] {$g$}
        edge [pre] node [label=above:{$\gamma_{g}^{2}$}] {} (7);
    \node [loc] (9) [below of=8] {$h$}
        edge [pre] node {} (8);

    \node [loc,inner sep=0.4pt] (10) [below of=7,yshift=-1mm] {$f$}
        edge [pre] node [label=right:{$\gamma_{f}^{2}$}] {} (7);

    \node [] (X1) [below of=10] {} edge [dotted,thick] node {} (10);

    \node [] (X2) [above of=5] {$s_g^1~~~~~~~$};
    \node [] (X3) [above of=9] {$s_g^2~~~~~~~$};

    \node [] (X4) [above of=7] {$~~~~~~~s_f^1$};
    \node [] (X5) [above of=X1] {$~~~~~~~s_f^2$};

    \node [] (X6) [right of=4] {$~~~~~~~s_1$};
    \node [] (X7) [right of=8] {$~~~~~~~s_2$};
\end{tikzpicture}} \\
        (a) & (c) \\
        \multicolumn{2}{c}{\input{fig_SETcountIf}} \\
        \multicolumn{2}{c}{(b)}
    \end{tabular}
    \caption{(a) A program with a function \texttt{countIf(A,n,x)}. (b) Symbolic execution tree of function \texttt{countIf}. (c) Compact symbolic execution tree of function \texttt{countIf}.
    }
    \label{fig:countIf}
\end{figure}

Let us now have a look at Figure~\ref{fig:countIf}~(a) depicting a program with a function \texttt{countIf}. The function counts number of elements in array \texttt{A} having values equal to \texttt{x}. We show the symbolic execution tree of the program at Figure~\ref{fig:countIf}~(b). There we can see several sequences of grey regions. According to our experience with the previous example we can easily detect that all that paths in grey are generated by a single program part consisting of locations $c,d,e,f,g$ and edges between them. But there are two cyclic paths $\pi = c,d,e,f,c$ and $\pi' = c,d,f,c$ inside the part. Nevertheless, the grey regions highlight only the cycle $\pi$. So, we ignore the cycle $\pi'$ and $\pi$ is therefore the only cycle we consider. The remainder is now obvious. The locations $f$ and $g$ are exits from the part and $c$ is the entry location into the part. The analysis of the path (discussed later in Section~\ref{sec:Detection}) computes the following template
\[t = (c, 2, \{ (\theta_f, \varphi_f, [], f)\prm{\kappa}, (\theta_g, \varphi_g, [], g)\prm{\kappa} \} )\]
Compact symbolic execution with the template $t$ computes compact symbolic execution tree depicted at Figure~\ref{fig:countIf}~(c). The tree is basically a single link list. Note that we instantiate the template each time we reach the location $c$. But for each such instantiation we need a fresh parameter to prevent parameter collisions from previous instantiations. We assume we have infinitely many different names for the parameters. Therefore, expressions and program states at Figure~\ref{fig:countIf}~(c) are as follows: $\gamma_g^i  = s_c^i.\theta\ese{t.\varphi_g\prm{\kappa_i}}$, $\gamma_f^i = s_c^i.\theta\ese{t.\varphi_f\prm{\kappa_i}}$, $s_g^i = s_c^i.\theta \circ (t.\theta_g, t.\varphi_g, [], g)\prm{\kappa_i}$ and $s_f^i = s_c^i.\theta \circ (t.\theta_f, t.\varphi_f, [], f)\prm{\kappa_i}$.

The sequences of grey regions in the tree at Figure~\ref{fig:countIf}~(b) goes bottom left. But imagine they would go bottom right. Then each region would represent a sequence of program locations $c,d,f,c$. If we analysed closer these sequences of grey regions, we would realise that there is a part of the program from Figure~\ref{fig:countIf}~(a) consisting of vertices $c,d,f,e,g$, where $c,d,f,c$ is the only cycle in the part, $c$ is the entry location into the part and locations $e$ and $g$ are exits from the part. If we further built a template from the part and run compact symbolic execution with it, we would also receive a compact symbolic execution tree forming basically a single linked list.

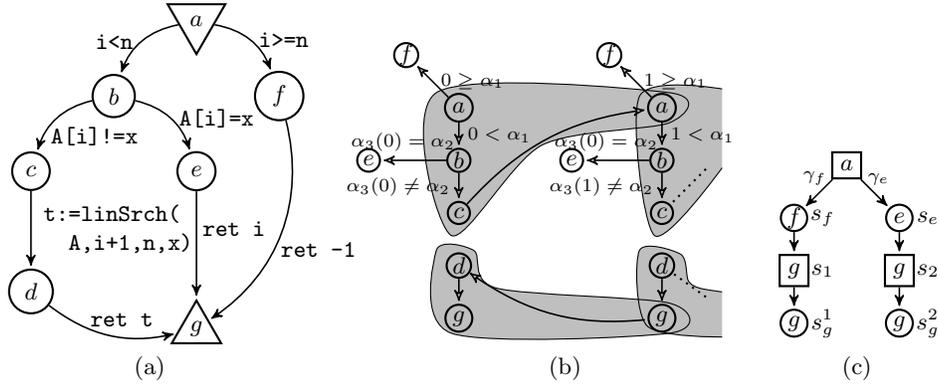
\begin{figure}[!htb]
    \hspace{-1.3cm}
    \begin{tabular}{ccc}
        {\centering
\tikzstyle{start} = [regular polygon,regular polygon sides=3,
regular polygon rotate=180,thick,draw,inner sep=2.0pt]
\tikzstyle{target} = [regular polygon,regular polygon sides=3,
regular polygon rotate=0,thick,draw,inner sep=1pt]
\tikzstyle{loc} = [circle,thick,draw]
\tikzstyle{pre} = [<-,shorten <=1pt,>=stealth',semithick]
\tikzstyle{post} = [->,shorten <=1pt,>=stealth',semithick]
\footnotesize
\begin{tikzpicture}[node distance=1.1cm]
    \node [start] (a) {$a$};
    \node [loc] (b) [left of=a,yshift=-10mm] {$b$}
        edge [pre,bend left] node [label=left:\texttt{i<n}] {} (a);
    \node [loc] (c) [left of=b,yshift=-10mm] {$c$}
        edge [pre,bend left]
                node [label=below:\texttt{~~~~~~A[i]!=x}] {} (b);
    \node [loc] (e) [right of=b,yshift=-10mm] {$e$}
        edge [pre,bend right] node [label=right:\texttt{A[i]=x}] {} (b);
    \node [loc] (d) [below of=c,yshift=-5mm] {$d$}
        edge [pre] node
            [label=right:{
                \begin{tabular}{l}
                    \texttt{t:=linSrch(} \\
                    \texttt{~~A,i+1,n,x)}
                \end{tabular}
            }] {} (c);
    \node [loc] (f) [right of=a,yshift=-10mm] {$f$}
        edge [pre,bend right] node [label=right:\texttt{i>=n}] {} (a);
    \node [target] (g) [below of=e,yshift=-10mm] {$g$}
        edge [pre,bend left=25]
                    node [label=above:\texttt{~~ret t}] {} (d)
        edge [pre] node [label=above:\texttt{~~~~~~ret i}] {} (e)
        edge [pre,bend right=35] node
            [label=below:\texttt{~~~~~~ret -1}] {} (f);
\end{tikzpicture}} & \hspace{-1.2cm}
        {\centering
\tikzstyle{start} = [regular polygon,regular polygon sides=3,
    regular polygon rotate=180,thick,draw,inner sep=0.1pt]
\tikzstyle{target} = [regular polygon,regular polygon sides=3,
    regular polygon rotate=0,thick,draw,inner sep=-0.5pt]
\tikzstyle{loc} = [circle,thick,draw,inner sep=0.5pt]
\tikzstyle{pre} = [<-,shorten <=1pt,>=stealth',semithick]
\tikzstyle{post} = [->,shorten <=1pt,>=stealth',semithick]
\footnotesize
\begin{tikzpicture}[node distance=0.7cm]


    \pgfsetmovetofirstplotpoint

    \pgfplothandlerrecord{\leftuppattern}
    \pgfplotstreamstart
        \pgfplotstreampoint{\pgfpoint{-0.3cm}{0.2cm}}
        \pgfplotstreampoint{\pgfpoint{2.8cm}{0.2cm}}
        \pgfplotstreampoint{\pgfpoint{2.8cm}{-0.2cm}}
        \pgfplotstreampoint{\pgfpoint{1.1cm}{-0.5cm}}
        \pgfplotstreampoint{\pgfpoint{-0.1cm}{-1.7cm}}
    \pgfplotstreamend

    \pgfplothandlerrecord{\rightuppattern}
    \pgfplotstreamstart
        \pgfplotstreampoint{\pgfpoint{3.5cm}{0.25cm}}
        \pgfplotstreampoint{\pgfpoint{2.4cm}{0.15cm}}
        \pgfplotstreampoint{\pgfpoint{2.6cm}{-1.7cm}}
        \pgfplotstreampoint{\pgfpoint{3.5cm}{-0.8cm}}
    \pgfplotstreamend

    \pgfplothandlerrecord{\leftdownpattern}
    \pgfplotstreamstart
        \pgfplotstreampoint{\pgfpoint{-0.3cm}{-1.9cm}}
        \pgfplotstreampoint{\pgfpoint{0.2cm}{-1.9cm}}
        \pgfplotstreampoint{\pgfpoint{0.5cm}{-2.4cm}}
        \pgfplotstreampoint{\pgfpoint{2.8cm}{-2.6cm}}
        \pgfplotstreampoint{\pgfpoint{2.8cm}{-3.0cm}}
        \pgfplotstreampoint{\pgfpoint{-0.1cm}{-3.0cm}}
    \pgfplotstreamend

    \pgfplothandlerrecord{\rightdownpattern}
    \pgfplotstreamstart
        \pgfplotstreampoint{\pgfpoint{3.5cm}{-2.45cm}}
        \pgfplotstreampoint{\pgfpoint{3.2cm}{-2.4cm}}
        \pgfplotstreampoint{\pgfpoint{3.05cm}{-2.3cm}}
        \pgfplotstreampoint{\pgfpoint{2.95cm}{-1.95cm}}
        \pgfplotstreampoint{\pgfpoint{2.4cm}{-1.9cm}}
        \pgfplotstreampoint{\pgfpoint{2.5cm}{-2.9cm}}
        \pgfplotstreampoint{\pgfpoint{3.5cm}{-3.05cm}}
    \pgfplotstreamend

    \pgfplothandlerclosedcurve\leftuppattern
        \pgfsetfillcolor{lightgray!100}\pgfusepath{fill}
    \pgfplothandlercurveto\rightuppattern
        \pgfsetfillcolor{lightgray!100}\pgfusepath{fill}
    \pgfplothandlerclosedcurve\leftdownpattern
        \pgfsetfillcolor{lightgray!100}\pgfusepath{fill}
    \pgfplothandlercurveto\rightdownpattern
        \pgfsetfillcolor{lightgray!100}\pgfusepath{fill}

    \pgfplothandlerclosedcurve\leftuppattern\pgfusepath{stroke}
    \pgfplothandlercurveto\rightuppattern\pgfusepath{stroke}
    \pgfplothandlerclosedcurve\leftdownpattern\pgfusepath{stroke}
    \pgfplothandlercurveto\rightdownpattern\pgfusepath{stroke}

    \node [loc,inner sep=1.5pt] (a1) {$a$};
    \node [loc,inner sep=-0.2pt] (f1) [left of=a1,above of=a1] {$f$}
        edge [pre] node [label=right:{\scriptsize
            $0\geq\alpha_1$}] {} (a1);
    \node [loc] (b1) [below of=a1] {$b$}
        edge [pre] node [label=right:{\scriptsize $0<\alpha_1$}] {} (a1);
    \node [loc,inner sep=1pt] (e1) [left of=b1,xshift=-5mm] {$e$}
        edge [pre] node [label=above:{\scriptsize
            $\alpha_3(0)=\alpha_2~~~$}] {} (b1);
    \node [loc,inner sep=1pt] (c1) [below of=b1] {$c$}
        edge [pre] node [label=left:{\scriptsize
            $\alpha_3(0)\neq\alpha_2$}] {} (b1);
    \node [loc] (d1) [below of=c1] {$d$};
    \node [loc,inner sep=1pt] (g1) [below of=d1] {$g$} edge [pre] node {} (d1);

    \node [loc,inner sep=1.5pt] (a2) [right of=a1,xshift=2cm] {$a$}
        edge [pre,bend right=20] node {} (c1);
    \node [loc,inner sep=-0.2pt] (f2) [left of=a2,above of=a2] {$f$}
        edge [pre] node [label=right:{\scriptsize
            $1\geq\alpha_1$}] {} (a2);
    \node [loc] (b2) [below of=a2] {$b$}
        edge [pre] node [label=right:{\scriptsize $1<\alpha_1$}] {} (a2);
    \node [loc,inner sep=1pt] (e2) [left of=b2,xshift=-5mm] {$e$}
        edge [pre] node [label=above:{\scriptsize
            $\alpha_3(0)=\alpha_2~~~$}] {} (b2);
    \node [loc,inner sep=1pt] (c2) [below of=b2] {$c$}
        edge [pre] node [label=left:{\scriptsize
            $\alpha_3(1)\neq\alpha_2$}] {} (b2);
    \node [loc] (d2) [below of=c2] {$d$};
    \node [loc,inner sep=1pt] (g2) [below of=d2] {$g$}
        edge [pre] node {} (d2)
        edge [post, bend left=20] node {} (d1);

    \node [] (X3) [right of=b2] {} edge [dotted,thick] node {} (c2);
    \node [] (X4) [right of=g2,yshift=2mm] {} edge [dotted,thick] node {} (d2);
\end{tikzpicture}} & \hspace{-0.6cm}
        {\centering
\tikzstyle{loc} = [circle,thick,draw,inner sep=0.5pt]
\tikzstyle{pre} = [<-,shorten <=1pt,>=stealth',semithick]
\tikzstyle{post} = [->,shorten <=1pt,>=stealth',semithick]
\footnotesize
\begin{tikzpicture}[node distance=0.7cm]
    \node [rectangle,thick,draw,inner sep=3pt] (1) {$a$};
    \node [loc,inner sep=0.1pt] (2) [left of=1,below of=1] {$f$}
        edge [pre] node [label=above:{\scriptsize $\gamma_f~$}] {} (1);
    \node [loc,inner sep=1.5pt] (3) [right of=1,below of=1] {$e$}
        edge [pre] node [label=above:{\scriptsize $~\gamma_e$}] {} (1);

    \node [rectangle,thick,draw] (4) [below of=2] {$g$}
        edge [pre] node {} (2);
    \node [loc,inner sep=1pt] (5) [below of=4] {$g$}
        edge [pre] node {} (4);

    \node [rectangle,thick,draw] (6) [below of=3] {$g$}
        edge [pre] node {} (3);
    \node [loc,inner sep=1pt] (7) [below of=6] {$g$}
        edge [pre] node {} (6);

    \node [] (X1) [above of=4] {$~~~~~~~s_f$};
    \node [] (X2) [above of=6] {$~~~~~~~s_e$};
    \node [] (X3) [below of=4] {$~~~~~~~s_g^1$};
    \node [] (X4) [below of=6] {$~~~~~~~s_g^2$};
    \node [] (X5) [above of=5] {$~~~~~~~s_1$};
    \node [] (X6) [above of=7] {$~~~~~~~s_2$};
\end{tikzpicture}}
        \\
        (a) & (b) & (c)
    \end{tabular}
    \caption{(a) A program with a recursive function \texttt{linSrch(A,i,n,x)}. (b) Symbolic execution tree of the recursive function \texttt{linSrch}. (c) Compact symbolic execution tree of the recursive function \texttt{linSrch}.
    }
    \label{fig:linSrchRec}
\end{figure}

Besides cyclic paths, recursive calls also produce real program paths with regularities in program states along them. At Figure~\ref{fig:linSrchRec}~(a) there is a recursive function \texttt{linSrchRec} which is equivalent to the function \texttt{linSrch} discussed before. Symbolic execution tree of the function is depicted at Figure~\ref{fig:linSrchRec}~(b). The root of the tree is the left-most vertex referencing program location $a$. There are two sequences of grey regions. The top sequence represents recursive calls, while the bottom sequence represents returning from the calls. We see that top sequence goes from left to the right. The bottom sequence goes in the opposite direction. We can further see there is one to one correspondence between regions of both sequences. Below each region in the top sequence, there is a single region of bottom sequence. Paths in both sequences of regions are connected in the tree. But this is not shown in the figure. The connection happens, when all the recursive calls are done and some basic case is executed in the recursive function. Then we get to the path of the bottom regions.

Let us first focus on the path at the top sequence of regions. Vertices in each region are related to the same sequence of program locations: $a,b,c,a$. Moreover, we enter the path in a vertex referencing location $a$ and we can leave the path either by stepping into a vertex referencing location $f$ or into a vertex referencing location $e$. If we look at the program (at Figure~\ref{fig:linSrchRec}~(a)), the sequence $a,b,c,a$ forms a cyclic path in it. Of course, the edge $(c,a)$ is not explicit in the program. But we consider it as a meta-edge labelled by an action simulating the effect of the function call, as defined by action of edge $(c,d)$. We now define a program part, say $P_1$, consisting of the cyclic path, the entry location $a$ and two exit locations $f$ and $e$. The part represent the phase of recursive calls of the function \texttt{linSrch}.

Now we similarly analyse the path in bottom sequence of regions. Each region repeats the same sequence of program locations $g,d,g$. The path is entered in vertex referencing location $g$, but there is no exit from the path. The sequence $g,d,g$ of locations forms a cyclic path in the program (at Figure~\ref{fig:linSrchRec}~(a)). Note, that we assume there is an artificial edge $(g,d)$ enclosing the cycle. Action of this edge is supposed to simulate the effect of return from the function call, as defined by action of edge $(c,d)$. We want to define a program part $P_2$ representing the phase of returning from recursive calls. We have the cyclic path and we have the entry location $g$ to the part. But there is no exit from the part. Obviously, the recursive calls ends in location $g$, where we leave the function. Therefore, our exit location is $g$ and we have the program part $P_2$. Note that if we want to formally match the exit location detection algorithm introduced for previous examples, we may imagine there is an edge from $g$ back to $g$ and labelled by \texttt{skip} action.

For the program parts $P_1$ and $P_2$ we compute the following templates $t_1$ and $t_1$ as described in the previous examples.
\begin{align*}
    t_1 &= (a, 2, \{ (\theta_f, \varphi_f, [], f)\prm{\kappa}, (\theta_e, \varphi_e, [], e)\prm{\kappa} \})
    \\
    t_2 &= (g, 1, \{ (\theta_g\prm{\kappa}, \true, [], g) \}).
\end{align*}
Note that the path condition of $t_2$ is simply $\true$, since we cannot escape from the path. In other words, as there is no branching along the path, the path condition cannot be updated from its initial value $\true$. It is important to note, that both templates use exactly the same parameter. The use of the same parameter creates a link between the number of recursive calls and number of returns from them. Having the templates we are able to formulate the template $t$ for the recursive function \texttt{linSrchRec}.
\[ t = (a, 2, \{ (\theta_f, \varphi_f, [], f)\prm{\kappa}, (\theta_e, \varphi_e, [], e)\prm{\kappa} \}, \theta_g\prm{\kappa}, g) \]
The template $t$ contains whole the template $t_1$, but it took only symbolic memory $\theta_g$ and the exit location $g$ from the template $t_2$.

We are ready to start compact symbolic execution with the template $t$. Symbolic execution tree for the program is depicted at Figure~\ref{fig:linSrchRec}~(c). First we step into the program location $a$. The tree contains only the root vertex referencing location $a$. The location $a$ is the entry location of $t$. Hence, we instantiate the first part of $t$ (related to phase of recursive calls, i.e.~related to $t_1$) into the tree. The number $2$ in $t$ identifies, that the root will have two successor vertices referencing locations $f$ and $e$ and they will be labelled by program states $s_f = (\theta_f, \varphi_f, [(t, \kappa)] \circ [], f)\prm{\kappa}$ and $s_e = (\theta_e, \varphi_e, [(t, \kappa)] \circ [], e)\prm{\kappa}$ respectively. Note that we omitted composition of these states with the initial program state labelling the root. We could do that, since composition of initial program state with any other state produces the other state again. Also note that call stacks of both states (i.e.~$[]$) are composed with a call stack containing a single special record of the form $(t, \kappa)$. This type of call stack record is introduced only for templates of recursive functions. First of all, this single record represents any number of subsequent recursive calls done by classic symbolic execution. And the record also saves reference to the template $t$ and the parameter $\kappa$ used in the instantiation. We note that edges from the root to its successors are labelled by expressions $\gamma_f \equiv t.\varphi_f\prm{\kappa}$ and $\gamma_e \equiv t.\varphi_e\prm{\kappa}$. Having computed successors of the root, we continue by classic symbolic execution independently from both of these vertices, until we reach the location $g$. For both the executions we do the same think at the location $g$. Let us consider execution continuing from the successor referencing location $f$. We need to instantiate the second part of the template representing returns from the recursive calls. So, we remove the record $(t, \kappa)$ from the top of the call stack, but we take the template $t$ and the parameter $\kappa$ stored in the record $(t, \kappa)$. In general, between both instantiation parts of a given template, there might be executed any code, there can be instantiated many other templates and there can even be instantiated the same template several times always with different (fresh) parameters. That is why we save the template and the parameter in the stack record. Let $s_1$ be a program state of the current leaf vertex of the tree. We create its only successor vertex labelled with program state $s_g^1 = (s_1.\theta \circ t.\theta\prm{\kappa_1}, s_1.\varphi, [], g)$. We see that there are two differences between states $s_1$ and $s_g^1$. First of all call stack of $s_g^1$ does not contain the special record $(t, \kappa)$ as we have popped it from the stack. And second, the symbolic memory of $s_g^1$ is the composition $s_1.\theta \circ t.\theta\prm{\kappa_1}$. Further classic symbolic execution form the vertex terminates, since we are leaving exit location of the starting function. We proceed similarly for the other run of symbolic execution (from the second successor of the root), where we get the final program state $s_g^2 = (s_2.\theta \circ t.\theta\prm{\kappa}, s_2.\varphi, [], g)$.

To summarise, a general scheme for compact symbolic execution of the examples above is as follows. We enumerate parts in a given program producing paths with regularities in program states along them. Such sources are mainly cyclic paths and pairs of cyclic paths representing recursion. For each enumerated part we compute a template. Then we run compact symbolic execution with the computed templates.

\section{Definition} \label{sec:Definition}

In this section, we give precise definition of templates parametrised by a single parameter. Templates for recursion consists of two parts instantiated independently into symbolic execution tree. These instances share the same parameter. We therefore show a process of information passing between different instances of the same template. And we further present compact symbolic execution algorithm using templates with one parameter with possible information exchange between instances. We start with basic terms valid for compact symbolic execution with any kinds of templates. We assume for the rest of this section that $P$ is a program.

An injective function $\Theta$ from a set of all program variables of $P$ to a set of \emph{symbols} $\{ \alpha_0, \alpha_1, \alpha_2, \ldots \}$ is an \emph{initial symbolic memory} of $P$. For each program variable \var{a} its symbol $\Theta(\var{a})$ represents some yet unknown value of that variable. So, $\Theta(\var{a})$ must belong to a domain of \var{a} (i.e.~$\Theta(\var{a})$ is of \var{a}'s type). Further, \emph{numeric symbolic expression} is application of operators to numeric constants and symbols. \emph{Boolean symbolic expression} is either an equality or inequality predicate over numeric symbolic expressions, or an application of logical connectives to other boolean symbolic expressions. \emph{Symbolic expression} is either numeric or boolean symbolic expression. We have already given the definition of symbolic memory, call stack and program state in Section~\ref{sec:Overview}. But we in addition define for any program state $s = (\theta, \varphi, \Xi, l)$ that $\theta(\var{a}) = \Theta(\var{a})$, for each local variable \var{a} undefined at location $l$. Also note that $\Theta$ is just a special symbolic memory.


The pseudo-code of Algorithm~\ref{alg:executeSymbolically} represents two algorithms. If we consider only unmarked lines, we get algorithm of classic symbolic execution. If we add lines marked with $\Box$ we get algorithm of compact symbolic execution with templates with a single parameter. The lines marked by $*$ are responsible for construction of symbolic execution tree. Obviously, both classic and compact symbolic executions can appear at both versions: with and without construction of the tree.

\begin{algorithm}[!h]
    \caption{\texttt{executeSymbolically}} \label{alg:executeSymbolically}
    \KwIn{$P$ - program to be executed \KwInSep
          $d$ - set of template detectors (only in $\Box$-version)}
    \KwOut{$E$ - set of final program states \KwOutSep
           $T$ - symbolic execution tree of $P$ (only in $*$-version)}
    \BlankLine
    \DontPrintSemicolon
    \lmark{$\Box$}
    Let $p$ be a set of all templates detected in $P$ by detectors $d$ \;
        \label{alg:executeSymbolically:DetectTemplates}
    $s_0 \asgn (\Theta, \true, [], $ entry location of the starting function$)$ \;
        \label{alg:executeSymbolically:InitialState}
    Le $Q$ be a queue of program states initially containing only $s_0$ \;
    \lmark{$*$}
    Create a root vertex of $T$ labelled with $s_0$ \;
        \label{alg:executeSymbolically:Root}
    \Repeat{\rm $Q$ becomes empty}{ \label{alg:executeSymbolically:LoopBegin}
        Extract the first program state $s$ from $Q$ \;
        \If{\rm $s.l$ is the exit location of the starting function or
                an error location}{
                \label{alg:executeSymbolically:FinalState}
            Insert $s$ into $E$ \;
        }
        \Else{
            $S \asgn \emptyset$ \;
            \lmark{$\Box$}
            \If(/* returning from recursion */)
            {\rm \texttt{top(}$s.\Xi$\texttt{)} $ = (t, \kappa)
                    \wedge s.l = s.\Xi.t.l'$}{
                    \label{alg:executeSymbolically:BoxesBegin}
                \lmark{$\Box$}
                $t \asgn s.\Xi.t$ \;
                    \label{alg:executeSymbolically:Unwinding}
                \lmark{$\Box$}
                $\kappa \asgn s.\Xi.\kappa$ \;
                    \label{alg:executeSymbolically:UnwindingGetKappa}
                \lmark{$\Box$}
                Replace all occurrences of the former parameter in
                    $t$ by $\kappa$\;
                \lmark{$\Box$}
                $s' \asgn (s.\theta \circ t.\theta\prm{\kappa}, s.\varphi,
                    \texttt{pop(}s.\Xi\texttt{)}, t.l')$ \;
                    \label{alg:executeSymbolically:SuccessorInUnwinding}
                    \label{alg:executeSymbolically:Pop}
                \lmark{$\Box$}
                Insert $s'$ into $S$\;
            }
            \lmark{$\Box$}
            \Else{
                \lmark{$\Box$}
                $p' \asgn \texttt{getTemplatesAt(}s.l, p\texttt{)}$ \;
                    \label{alg:executeSymbolically:GetTemplatesAt}
                \lmark{$\Box$}
                \If {$p' \neq \emptyset$}{
                    \lmark{$\Box$}
                    $t \asgn \texttt{chooseTemplate(}p'\texttt{)}$ \;
                        \label{alg:executeSymbolically:ChooseTemplate}
                    \lmark{$\Box$}
                    $\kappa \asgn \texttt{getFreshParam()}$ \;
                        \label{alg:executeSymbolically:FreshKappa}
                    \lmark{$\Box$}
                    Replace all occurrences of the former parameter in
                        $t$ by $\kappa$\;
                    \lmark{$\Box$}
                    \If(/* recursive calling */)
                    {\rm $t$ is a recursion template}{
                        \lmark{$\Box$}
                        \ForEach{$i = 1, \ldots, t.n$}{
                                \label{alg:executeSymbolically:Calling}
                            \lmark{$\Box$}
                            $s' \asgn s \circ (t.\theta_i, t.\varphi_i,
                                    [(t,\kappa)] \circ t.\Xi_i,
                                    t.l_i)\prm{\kappa}$\;
                              \label{alg:executeSymbolically:SuccessorInCalling}
                              \label{alg:executeSymbolically:Push}
                            \lmark{$\Box$}
                            Insert $s'$ into $S$\;
                        }
                    }
                    \lmark{$\Box$}
                    \Else(/* $t$ is a general template with one parameter */){
                        \lmark{$\Box$}
                        \ForEach{$i = 1, \ldots, t.n$}{
                                \label{alg:executeSymbolically:Unrolling}
                            \lmark{$\Box$}
                            $s' \asgn s \circ (t.\theta_i, t.\varphi_i,
                                               t.\Xi_i, t.l_i)\prm{\kappa}$\;
                            \label{alg:executeSymbolically:SuccessorInUnrolling}
                            \lmark{$\Box$}
                            Insert $s'$ into $S$\;
                        }
                    }
                }
                \lmark{$\Box$}
                \Else(/* applying classic symbolic execution step */){
                        \label{alg:executeSymbolically:BoxesEnd}
                    $S \asgn \texttt{computeClassicSuccessors(}P,s\texttt{)}$ \;
                        \label{alg:executeSymbolically:ClassicStep}
                }
            }
            \lmark{$*$}
            Let $u$ be a leaf of $T$ whose label is $s$ \;
                    \label{alg:executeSymbolically:GetLeaf}
            \ForEach{\rm program state $s' \in S$ such that
                            $s'.\varphi$ is satisfiable}{
                    \label{alg:executeSymbolically:AddSATSuccessors}
                Insert $s'$ at the end of $Q$ \;
                \lmark{$*$}
                Insert a new vertex $v$ labeled with $s'$ into $T$ \;
                    \label{alg:executeSymbolically:CreateVertex}
                \lmark{$*$}
                Insert an edge $(u,v)$ into $T$ \;
                    \label{alg:executeSymbolically:InsertEdge}
            }
        }
    } \label{alg:executeSymbolically:LoopEnd}
    \Return{$E$ \\ \lmark{$*$} \hspace{1.05cm} $T$} \;
\end{algorithm}

We now describe the algorithm of classic symbolic execution. At line~\ref{alg:executeSymbolically:InitialState}, there we create initial program state and then we insert it into a queue $Q$. The queue $Q$ keeps all program states for which we have not been computing successor program states yet. Until $Q$ becomes empty, we iterate the loop at lines~\ref{alg:executeSymbolically:LoopBegin}--\ref{alg:executeSymbolically:LoopEnd}. At line~\ref{alg:executeSymbolically:FinalState} we detect whether actually processed program state $s$ is final or not. If it is not, we compute its successors at line~\ref{alg:executeSymbolically:ClassicStep}. In short, the function \texttt{computeClassicSuccessors} either executes actions of out-edges from location $s.l$ or it resolves return from a call, if $s.l$ is a function exit location. We already gave an intuition how to symbolically execute actions at the beginning of Section~\ref{sec:Overview}. We further see at line~\ref{alg:executeSymbolically:AddSATSuccessors} that we discard all successors of $s$, whose path conditions are not satisfiable. Discarded states do not represent real behaviour of the program.

Now we focus on $*$-version of the algorithm. We create root of the tree labelled by the initial program state at line~\ref{alg:executeSymbolically:Root}. When processing a state $s$ inside the loop we take the only leaf in the tree labelled with $s$ at line~\ref{alg:executeSymbolically:GetLeaf}. We compute its successor vertices at lines~\ref{alg:executeSymbolically:CreateVertex} and \ref{alg:executeSymbolically:InsertEdge}. Note that the successors are labelled by successor states of $s$.

We have to postpone description of $\Box$-version of the algorithm, until we have properly defined templates the algorithm uses. The first step toward the definition is introduction of parameters and their substitution.


We distinguish a set $\{ \kappa, \tau, \kappa_1, \tau_1, \kappa_2, \tau_2, \ldots \}$ of variables called \emph{parameters}, ranging over non-negative integers. We extend numeric symbolic expressions such that they may also contain application of operators to parameters. We allow boolean symbolic expression to contain quantification of parameters. We further naturally extend symbolic memories, call stacks and program states to contain symbolic expressions with parameters. When we want to emphasise that $\bld{\kappa}$ is a set of all parameters appearing in a symbolic expression $\varphi$, we denote it as $\varphi\prm{\bld{\kappa}}$. And if we want to emphasise that a symbolic expression $\varphi$ does not contain any parameter, we denote it as $\varphi\prm{}$. We naturally extend the notations above for symbolic memories, call stacks and program states.

We now describe substitution of parameters. Each function from a finite set of parameters to non-negative integers is \emph{valuation}. Let $\varphi\prm{\bld{\kappa}}$, $\theta\prm{\bld{\kappa}}$, $\Xi\prm{\bld{\kappa}}$ and $s\prm{\bld{\kappa}}$ be a symbolic expression, a symbolic memory, a call stack and a program state respectively, $\bld{\kappa} \neq \emptyset$ and $\bld{\nu}$ be a valuation defined for all parameters in $\bld{\kappa}$. Then we compute $\varphi\prm{\bld{\nu}}$ from $\varphi\prm{\bld{\kappa}}$ such that we substitute all parameters in $\varphi$ by related integers in $\bld{\nu}$. We compute $\theta\prm{\bld{\nu}}$ from $\theta\prm{\bld{\kappa}}$ such that we substitute all parameters in all the expressions in $\theta$ by related integers in $\bld{\nu}$. Substitution of call stack parameter is a bit more complicated,
since we introduced the special form $(t, \kappa)$ of a stack record in the last example of Section~\ref{sec:Overview}. Therefore, to prepare ground for stack equivalence, we compute $\Xi\prm{\bld{\nu}}$ from $\Xi\prm{\bld{\kappa}}$ in the following two steps: (1) We update each record $(\sigma, l)$ of the call stack $\Xi$ to $(\sigma\prm{\bld{\nu}}, l)$ (note that $\sigma$ is basically symbolic memory, only restricted to local variables). (2) Each record of the special form $(t, \kappa)$  in the call stack form the previous step is either discarded, if $\bld{\nu}(\kappa) = 0$, or it is replaced by $\bld{\nu}(\kappa)$ records $(\bot,\bot)$, where symbol $\bot$ represent \emph{any} possible content. Therefore, the record $(\bot,\bot)$ represents any possible stack record (the first $\bot$ in the record represents any possible content of $\sigma$ and the second $\bot$ represents any possible program location).

We often use the following simplified notation. If an expression $\varphi$ contains exactly one parameter $\kappa$ and a $\{ (\kappa,\nu) \}$ is a valuation, then we write $\varphi\prm{\kappa}$ and $\varphi\prm{\nu}$ instead of $\varphi\prm{\{\kappa\}}$ and $\varphi\prm{\{(\kappa,\nu)\}}$ respectively. The notation also applies to symbolic memories, call stacks and program states.

Next we define composition of program states and equivalence between them. We also express some basic equivalences for compositions.

\begin{dfn}[Composition]
    Let $\Xi = [r_0, \ldots, r_m]$ and $\Xi' = [r_0', \ldots, r_n']$ be call stacks and $s = (\theta, \varphi, \Xi, l)$ and $s' = (\theta', \varphi', \Xi', l')$ be program states. Then composite program state $s \circ s' = (\theta \circ \theta',~\varphi \wedge \theta\ese{\varphi'}, \Xi \circ (\theta \circ \Xi'), l')$, where $\theta\ese{\varphi'}$ is a symbolic expression constructed from $\varphi'$ such that all symbols $\alpha_i$ in $\varphi'$ are simultaneously substituted by symbolic expressions $\theta(\Theta^{-1}(\alpha_i))$, $\theta \circ \theta'$ is a symbolic memory such that for each variable \var{a} we have $(\theta \circ \theta')(\var{a}) = \theta\ese{\theta'(\var{a})}$, $\theta \circ \Xi' = [\bar{r}_0', \ldots, \bar{r}_n']$, where each $\bar{r}_i'$ is equal to $r_i$ except the first component being $\bar{r}_i'.\sigma = \theta \circ r_i'.\sigma$, and $\Xi \circ (\theta \circ \Xi') = [r_0, \ldots,$ $ r_m, \bar{r}_0', \ldots, \bar{r}_n']$, .
\end{dfn}

\begin{dfn}[Equivalence]
    Let $\varphi$, $\varphi'$ be symbolic expressions, $\theta$, $\theta'$ be symbolic memories, $\Xi = [r_0, \ldots, r_m]$, $\Xi' = [r_0', \ldots, r_n']$ be call stacks and $s$, $s'$ be program states. Then $\varphi \equiv \varphi'$, if $\varphi$ and $\varphi'$ are either logically equivalent boolean symbolic expressions or numeric symbolic expressions such that $(\varphi = \varphi') \equiv \true$. $\theta \equiv \theta'$, if for each variable \var{a} we have $\theta(\var{a}) \equiv \theta'(\var{a})$. $\Xi \equiv \Xi'$, if $m = n$ and for each $i \in \{0, \ldots, m\}$ we have $r_i.\sigma$ and $r_i'.\sigma$ are defined for the same variables with equivalent values and $r_i.l = r_i'.l$. And $s \equiv s'$, if both $s$ and $s'$ have equal or equivalent components.
\end{dfn}

When returning from a function call, values of local variables are discarded. Therefore, if we have two program states at the same exit location of a function, we may restrict equivalence between symbolic memories of these states only to global variables. Therefore we define also the following equivalence between program states.

\begin{dfn}[Equivalence on Global Variables]
    Let $s$ and $s'$ be program states. Then $s$ is equivalent on global variables with $s'$, written by $s \overset{g}{\equiv} s'$, if they have equal or equivalent components except one with symbolic memories, where for each global variable \var{a} we require $s.\theta(\var{a}) \equiv s'.\theta(\var{a})$.
\end{dfn}

We summarise basic equivalences between composed program states in the following lemma. We do not provide proof since the equivalences is mostly obvious or easy to check.

\begin{lem}[Equivalent Compositions]
    Let $s, s'$ and $s''$ be program states, $\bld{\nu}$ and $\bld{\nu}'$ be valuations of all parameters in $s$ and $s'$ respectively such that $\bld{\nu} \cup \bld{\nu}'$ is also a valuation, $\theta$, $\theta'$ and $\theta''$ be symbolic memories and $\varphi$ and $\psi \wedge \psi'$ be symbolic expressions. Then $s \circ (s' \circ s'') \equiv (s \circ s') \circ s''$, $s\prm{\bld{\nu}} \circ s'\prm{\bld{\nu}'} \equiv (s \circ s')\prm{\bld{\nu} \cup \bld{\nu}'}$, $\theta \circ (\theta' \circ \theta'') \equiv (\theta \circ \theta') \circ \theta''$, $(\theta \circ \theta')\ese{\varphi} \equiv \theta\ese{\theta'\ese{\varphi}}$ and $\theta\ese{\psi} \wedge \theta\ese{\psi'} \equiv \theta\ese{\psi \wedge \psi'}$.
\end{lem}

Before we formulate a definition of templates with one parameter we give its intuition. Let us consider a part of the program $P$ with an entry location $e$ and $n$ distinct exit locations $x_1, \ldots, x_n$. We saw in Section~\ref{sec:Overview}, that key properties for building a template of the part are program states $s_1\prm{\kappa}, \ldots, s_n\prm{\kappa}$ at exit locations $x_1, \ldots, x_n$. We need to ensure that states $s_i\prm{\kappa}$ correctly represent behaviour of the analysed part. King proved~\cite{Kin76} that path conditions at leaf vertices of symbolic execution tree $T$ of $P$ are satisfiable. Therefore, if $s_i.\varphi$ is not satisfiable, then there cannot be a path in $T$ traversing the part form $e$ to $x_i$. The exit $x_i$ is thus useless for the construction of the template and we omit it. King further showed~\cite{Kin76} that for two different leaf vertices $u$ and $v$ of $T$ we have $u.\varphi \wedge v.\varphi \equiv \false$. This statement is also valid for program parts. So, we require $(s_i.\varphi \wedge s_j.\varphi) \equiv \false$ for all different $i$ and $j$. We summarise these requirements in the following definition.

\begin{dfn}[Templates with one parameter] \label{def:Templates}
    Let $T$ be symbolic execution tree of $P$ computed by $*$-version of Algorithm~\ref{alg:executeSymbolically}, $n>0$ be an integer, $l, l', l_1, \ldots, l_n$ be locations in $P$, $\kappa$ be a parameter, $\theta\prm{\kappa}, \theta_1\prm{\kappa}, \ldots,$ $\theta_n\prm{\kappa}$ be symbolic memories $\varphi_1\prm{\kappa}, \ldots, \varphi_n\prm{\kappa}$ be satisfiable boolean symbolic expressions such that for each $i,j \in \{ 1, \ldots, n \}, i \neq j$ we have $(\varphi_i \wedge \varphi_j) \equiv \false$ and let $\Xi_1\prm{\kappa}, \ldots, \Xi_n\prm{\kappa}$ be call stacks.

    A tuple $t = (l, n, \{ (\theta_1, \varphi_1, \Xi_1, l_1), \ldots, (\theta_n,\varphi_n, \Xi_n, l_n) \})$ is a \emph{template with one parameter $\kappa$} in $P$, if
    \begin{itemize}
        \item[(L1)] All the locations $l, l_1, \ldots, l_n$ in $t$ are neither entry nor exit ones.

        \item[(L2)] For each path $\pi = u \omega$ in $T$ from any vertex $u$ satisfying $u.l = t.l$ to a leaf, there is a vertex $w \in \omega$, an index $i \in \{ 1, \ldots, n \}$ and an integer $\nu \geq 0$, such that $w.s \equiv u.s \circ (t.\theta_i,~t.\varphi_i,~t.\Xi_i,~t.l_i)\prm{\nu}$.

        \item[(L3)] For each vertex $u$ of $T$, an index $i \in \{ 1, \ldots, n \}$  and non-negative integer $\nu$ such that $u.l = t.l$ and $(u.\varphi \wedge u.\theta\ese{t.\varphi_i\prm{\nu}})$ is satisfiable, there is a successor $w$ of $u$ in $T$ such that $w.s \equiv u.s \circ (t.\theta_i,~t.\varphi_i,~t.\Xi_i,~t.l_i)\prm{\nu}$.
    \end{itemize}

    A tuple $t = (l, n, \{ (\theta_1, \varphi_1, \Xi_1, l_1), \ldots, (\theta_n,\varphi_n, \Xi_n, l_n) \}, \theta, l')$ is a  \emph{recursion template with one parameter $\kappa$} in $P$, if
    \begin{itemize}
        \item[(R1)] $t.l$ and $t.l'$ are entry and exit locations of the same function respectively and $t.l'$ is the target vertex of an edge with a call action of that function. All the locations $l_1, \ldots, l_n$ in $t$ are neither entry nor exit ones.

        \item[(R2)] For each path $\pi = u \omega$ in $T$ from any vertex $u$ satisfying $u.l = t.l$ to a leaf, there is a non-leaf vertex $w \in \omega$, an index $i \in \{ 1, \ldots, n \}$ and an integer $\nu \geq 0$, such that $w.s \equiv u.s \circ (t.\theta_i,~t.\varphi_i,~[(t,\kappa)] \circ t.\Xi_i,~t.l_i)\prm{\nu}$.

        Further, if there is the first successor $\bar{w}$ of $w$ in $\pi$ such that $\bar{w}.l = t.l'$ and $\bar{w}.\Xi = w.\Xi$, then there is a non-leaf vertex $\bar{u}$ in a suffix of $\pi$ starting with $\bar{w}$ such that $\bar{u}.s \overset{g}{\equiv} (\bar{w}.\theta \circ t.\theta\prm{\nu}, \bar{w}.\varphi, u.\Xi, t.l')$.

        \item[(R3)] For each vertex $u$ of $T$, an index $i \in \{ 1, \ldots, n \}$  and non-negative integer $\nu$ such that $u.l = t.l$ and $(u.\varphi \wedge u.\theta\ese{t.\varphi_i\prm{\nu}})$ is satisfiable, there is a successor $w$ of $u$ in $T$ such that $w.s \equiv u.s \circ (t.\theta_i,~t.\varphi_i,~[(t,\kappa)] \circ t.\Xi_i,~t.l_i)\prm{\nu}$.
    \end{itemize}
\end{dfn}

Note that requirements (L2) and (R2) guarantees that no path in $T$ with vertices $u$ and $v$ such that $u.l = l$ and $v.l = l_i$ is suppressed by the state $(t.\theta_i,~t.\varphi_i,~t.\Xi_i,~t.l_i)\prm{\kappa}$. And requirement (L3) and (R3) guarantees that program state $(t.\theta_i,~t.\varphi_i,~t.\Xi_i,~t.l_i)\prm{\kappa}$ does not produce unreal paths. Also note that in requirement (R2) there we use restriction of equivalence to global variables for the phase of returning from recursive calls. Since values of local variables are not important when returning from a function call, the restriction may help to simplify detection of a recursion template.

We are ready to describe $\Box$-version at Algorithm~\ref{alg:executeSymbolically}.
%
%
At line~\ref{alg:executeSymbolically:DetectTemplates} there we detect templates with one parameter in the passed program $P$. That is a task for so called \emph{template detectors}. We discuss a possible construction of such a detector in Section~\ref{sec:Detection}. The only purpose of  lines~\ref{alg:executeSymbolically:BoxesBegin}--\ref{alg:executeSymbolically:BoxesEnd} is to compute successor states of a currently processed program state $s$. Let us first assume the test at line~\ref{alg:executeSymbolically:BoxesBegin} is $\false$. So, we get to line~\ref{alg:executeSymbolically:GetTemplatesAt}. There we call a system function \texttt{getTemplatesAt}, which selects those templates, whose entry locations matches the actual program location $s.l$. If the selection is not empty we may instantiate one of the selected templates. A system function \texttt{chooseTemplate} is supposed to choose exactly one template $t$ to be instantiated. We may for example choose randomly. We do not put any constraints to the selection strategy. To prevent parameter collisions we first get a fresh one at line~\ref{alg:executeSymbolically:FreshKappa} and then we replace the parameter used in $t$ by default by the fresh one. Now we have two possibilities. Either $t$ is a recursion template or not. In the first case we get to a loop at line~\ref{alg:executeSymbolically:Calling}. There we create $t.n$ successors of the program state $s$ (see line~\ref{alg:executeSymbolically:SuccessorInCalling}). Note that call stack of $i$-th successor state is of the form $s.\Xi \circ [(t, \kappa)] \circ t.\Xi_i$. It means that the special record is at the position in the stack, when we entered the recursive function. The only special record $(t, \kappa)$ in the call stack represents \emph{any} possible number of subsequent recursive calls in classic symbolic execution. The record also saves reference to the template $t$ and the parameter $\kappa$ for the later phase of returning from the recursive calls. If $t$ is not a recursion template, then it must be our general purpose template with one parameter (since we do not consider any other kinds of templates in this paper). So we get to line~\ref{alg:executeSymbolically:Unrolling} in the algorithm. There we also create successors of the program state $s$ (see line~\ref{alg:executeSymbolically:SuccessorInUnrolling}). It remains to discuss the computation of successors, when the condition at line~\ref{alg:executeSymbolically:BoxesBegin} is $\true$. The condition says that the location $s.l$ references exit location of a function and that there is the special record $(t, \kappa)$ at the top of the call stack $s.\Xi$. In other words, we reached the moment, when we have to return from recursive calls. We first retrieve the recursive template and the parameter used in the instantiation of $t$ (see lines~\ref{alg:executeSymbolically:Unwinding} and \ref{alg:executeSymbolically:UnwindingGetKappa}). After substitution of the default parameter by the retrieved one, we finish the instantiation of $t$ by computing the only successor of the actual state. The successor state represents the effect of all the returns from recursive calls done previously. This is ensured by using of the same parameter form both phases of the instantiation of the template $t$. A number of recursive calls therefore matches the number of returns form them. Also note that call stack of the successor does not contain the special record. We finish the description of the algorithm by the following observation. The expressions computing successor states at lines~\ref{alg:executeSymbolically:SuccessorInUnwinding}, \ref{alg:executeSymbolically:SuccessorInCalling} and \ref{alg:executeSymbolically:SuccessorInUnrolling} precisely match corresponding expressions in Definition~\ref{def:Templates}. Note that at line~\ref{alg:executeSymbolically:SuccessorInUnwinding} there the call stack \texttt{pop(}$s.\Xi$\texttt{)} must be equal to one of a program state, for which we previously get to line~\ref{alg:executeSymbolically:SuccessorInCalling}. And this program state had to be related to the entry location of a function causing the recursive calls.

\section{Soundness and Completeness} \label{sec:SoundnessCompleteness}

In this section we formulate and prove soundness and completeness theorems for compact symbolic execution using recursive and general templates with one parameter. The theorems say that both classic and compact symbolic execution explore the same set of real paths of $P$. To avoid repetitions we assume for the remainder of this section that $P$ is a program, and $T$ and $T'$ are symbolic execution trees of the program $P$ computed by $*$- and $\Box,*$-versions of Algorithm~\ref{alg:executeSymbolically} respectively.

\begin{lem} \label{lem:NonAdjacentPushes}
    Call stack records pushed at line~\ref{alg:executeSymbolically:Push} of Algorithm~\ref{alg:executeSymbolically} cannot be adjacent in call stacks of vertices of $T'$.
\end{lem}
\begin{proof}
    Follows immediately from requirement for locations of templates in Definition~\ref{def:Templates} and from the fact, that reaching line~\ref{alg:executeSymbolically:Push} requires a processed state must reference a function entry location.
\end{proof}

\begin{lem} \label{lem:GEquivToEquiv}
    Let $u \in T$, $u' \in T'$, $u'.\Xi \neq []$, {\tt top(}$u'.\Xi${\tt )} $ = (t, \kappa)$, $u'.l$ is an exit location and $u.s \overset{g}{\equiv} u'.s\prm{\bld{\nu}}$ for some valuation $\bld{\nu}$. Then there are the only direct successors $w \in T$ and $w' \in T'$ of $u$ and $u'$ respectively and they satisfy $w.s \equiv w'.s\prm{\bld{\nu}}$.
\end{lem}
\begin{proof}
    Follows directly from Lemma~\ref{lem:NonAdjacentPushes} and from the fact that successors of $u'$ are computed at line~\ref{alg:executeSymbolically:ClassicStep} of Algorithm~\ref{alg:executeSymbolically}.
\end{proof}

\begin{thm}[Soundness]
    For each leaf vertex $e \in T$ there is a leaf vertex $e' \in T'$ and a valuation  $\bld{\nu}$ of all parameters in $e'.s$ such that $e.s \equiv e'.s\prm{\bld{\nu}}$.
\end{thm}
\begin{proof}
    Let $\pi$ be the path in $T$ from the root to the leaf vertex $e$. We prove the theorem by the following induction:

    \emph{Basic case}: The root vertices $r$ and $r'$ of $T$ and $T'$ respectively are labelled by the same program state $s_0$ (see lines \ref{alg:executeSymbolically:InitialState} and \ref{alg:executeSymbolically:Root}). So, $r.s \equiv r'.s\prm{\bld{\nu}}$, for $\bld{\nu} = \emptyset$.

    \emph{Inductive step}: Let $u \in \pi$, $u \neq e$, $u'$ be a vertex of $T'$ and $\bld{\nu}$ be a valuation such that $u.s \equiv u'.s\prm{\bld{\nu}}$. We show, there is a successor $w$ of $u$ in $\pi$, a successor vertex $w'$ of $u'$ in $T'$ and a valuation $\bld{\nu}'$ such that $w.s \equiv w'.s\prm{\bld{\nu}'}$. And we further show there is no vertex $v'$ in the path between $u'$ and $w'$ in $T'$ such that successors of $v'.s$ are computed at line~\ref{alg:executeSymbolically:SuccessorInCalling}. There are four possible cases in Algorithm~\ref{alg:executeSymbolically} for $u'.s$:

    (1) We reach line~\ref{alg:executeSymbolically:Unrolling}: According to Definition~\ref{def:Templates}~(L2), there is a successor vertex $w$ of $u$ in $\pi$, an index $i$ and a non-negative integer $\nu$ for $\kappa$ such that
    \begin{align*}
        w.s &\equiv u.s \circ (t.\theta_i\prm{\kappa}, t.\varphi_i\prm{\kappa}, t.\Xi_i\prm{\kappa}, t.l_i)\prm{\{(\kappa,\nu)\}}
        \\
        &\equiv u'.s\prm{\bld{\nu}} \circ (t.\theta_i\prm{\kappa}, t.\varphi_i\prm{\kappa}, t.\Xi_i\prm{\kappa}, t.l_i)\prm{\{(\kappa,\nu)\}}
        \\
        &\equiv (u'.s \circ (t.\theta_i\prm{\kappa}, t.\varphi_i\prm{\kappa}, t.\Xi_i\prm{\kappa}, t.l_i))\prm{\bld{\nu} \cup \{(\kappa,\nu)\} }
        \\
        &\equiv s'\prm{\bld{\nu}'},
    \end{align*}
    where $s'$ is the $i$-th direct successor of $u'.s$ computed at line~\ref{alg:executeSymbolically:SuccessorInUnrolling}. And since $w \in T$, we have $w'.\varphi$ is satisfiable. Therefore, there is be a direct successor $w'$ of $u'$ in $T'$ with $w.s = s'$.

    (2) We reach line~\ref{alg:executeSymbolically:Calling}: According to Definition~\ref{def:Templates}~(R2), there is a successor vertex $w$ of $u$ in $\pi$, an index $i$ and a non-negative integer $\nu$ for $\kappa$ such that
    \begin{align*}
        w.s &\equiv u.s \circ (t.\theta_i\prm{\kappa}, t.\varphi_i\prm{\kappa}, [(t,\kappa)] \circ t.\Xi_i\prm{\kappa}, t.l_i)\prm{\{(\kappa,\nu)\}}
        \\
        &\equiv u'.s\prm{\bld{\nu}} \circ (t.\theta_i\prm{\kappa}, t.\varphi_i\prm{\kappa}, [(t,\kappa)] \circ t.\Xi_i\prm{\kappa}, t.l_i)\prm{\{(\kappa,\nu)\}}
        \\
        &\equiv (u'.s \circ (t.\theta_i\prm{\kappa}, t.\varphi_i\prm{\kappa}, [(t,\kappa)] \circ t.\Xi_i\prm{\kappa}, t.l_i))\prm{\bld{\nu} \cup \{(\kappa,\nu)\} },
        \\
        &\equiv s'\prm{\bld{\nu}'},
    \end{align*}
    where $s'$ is the $i$-th direct successor of $u'.s$ computed at line~\ref{alg:executeSymbolically:SuccessorInCalling}. And since $w \in T$, we have $w'.\varphi$ is satisfiable. Therefore, there is a direct successor $w'$ of $u'$ in $T'$ with $w.s = s'$.

    (3) We reach line~\ref{alg:executeSymbolically:Unwinding}: 
    Let $\pi'$ be a path in $T'$ from the root to the vertex $u'$. According to connections between vertices $u'$ constructed for vertices $u$ along $\pi$, there is a predecessor $x'$ of $u'$ in $\pi'$, which pushed (at line~\ref{alg:executeSymbolically:Push}) the record being at the top of $u'.\Xi$. Obviously, successors of $x'.s$ are computed at line~\ref{alg:executeSymbolically:SuccessorInCalling}. Therefore, there is $x \in \pi$ such that $x.s \equiv x'.s\prm{\bld{\nu}}$. According to case (2) there is a successor $y$ of $x$ in $\pi$ and a direct successor $y'$ of $x'$ in $\pi'$ such that $y.s \equiv y'.s\prm{\bld{\nu}}$. Note that $y'.s$ uses the parameter $\kappa$ retrieved from stack $u'.\Xi$ at line~\ref{alg:executeSymbolically:UnwindingGetKappa}. Therefore, valuation $\bld{\nu}$ defines an integer $\nu = \bld{\nu}(\kappa)$. Also note that $u$ is the first successor of $y$ in $\pi$ with $u.l$ being an exit location and $u.\Xi = y.\Xi$. Otherwise we would apply this case (3) for some other vertex lying between $y'$ and $u'$ in $\pi'$. Therefore, from Definition~\ref{def:Templates}~(R2) there is a non-leaf vertex $v$ in a suffix of $\pi$ starting with $u$ such that
    \begin{align*}
        v.s &\overset{g}{\equiv} (u.\theta \circ t.\theta\prm{\kappa}, u.\varphi, x.\Xi, t.l')\prm{\{(\kappa,\nu)\}}
        \\
        &\overset{g}{\equiv} (u'.\theta\prm{\bld{\nu}} \circ t.\theta\prm{\kappa}, u'.\varphi\prm{\bld{\nu}}, \texttt{pop(}u'.\Xi\texttt{)}\prm{\bld{\nu}}, t.l')\prm{\{(\kappa,\nu)\}}
        \\
        &\overset{g}{\equiv} (u'.\theta \circ t.\theta\prm{\kappa}, u'.\varphi, \texttt{pop(}u'.\Xi\texttt{)}, t.l')\prm{\bld{\nu}}
        \\
        &\overset{g}{\equiv} s'\prm{\bld{\nu}},
    \end{align*}
    where $s'$ is the only successor state of $u'.s$ computed at line~\ref{alg:executeSymbolically:SuccessorInUnwinding}. Since $v \in T$, then $s'.\varphi$ is satisfiable and there is a direct successor $v'$ of $u'$ in $T'$ with $v'.s = s'$. And finally Lemma~\ref{lem:GEquivToEquiv} ensures there are the only direct successors $w$ and $w'$ of $v$ and $v'$ respectively, such that $w.s \equiv w'.s\prm{\bld{\nu}}$.

    (4) Otherwise, we reach line~\ref{alg:executeSymbolically:ClassicStep}: Since $u.s \equiv u'.s\prm{\bld{\nu}}$ and we apply classic symbolic execution step for $u'.s$, there must be a direct successor $w$ of $u$ and a direct successor $w'$ of $u'$ such that $w.s \equiv w'.s\prm{\bld{\nu}}$.
\end{proof}

\begin{thm}[Completeness]
    For each leaf vertex $e' \in T'$ there is a leaf vertex $e \in T$ and a valuation $\bld{\nu}$ of all parameters in $e'.s$ such that $e.s = e'.s\prm{\bld{\nu}}$.
\end{thm}
\begin{proof}
    Let $\pi'$ be the path in $T'$ from the root to the leaf vertex $e'$. We prove the theorem by the following induction:

    \emph{Basic case}: The root vertices $r$ and $r'$ of $T$ and $T'$ respectively are labelled by the same program state $s_0$ (see lines \ref{alg:executeSymbolically:InitialState} and \ref{alg:executeSymbolically:Root}). Let us construct a non-empty set $U$ of vertices of $T$ such that for each valuation $\bld{\nu}$ of all parameters in $r'.s$ such that $r'.\varphi\prm{\bld{\nu}}$ is satisfiable, there is $u \in U$ such that $u.s \equiv r'.s\prm{\bld{\nu}}$. Obviously $U = \{ r \}$, because $r'.\varphi$ contains no parameter (so $r.s \equiv r'.s\prm{\bld{\nu}}$, for each $\bld{\nu})$.

    \emph{Inductive step}:  Let $u' \in \pi'$, $u' \neq e'$ and $U$ be a non-empty set of vertices of $T$ such that for each valuation $\bld{\nu}$ of all parameters in $u'.s$ such that $u'.\varphi\prm{\bld{\nu}}$ is satisfiable, there is $u \in U$ such that $u.s \equiv u'.s\prm{\bld{\nu}}$. We show, there is a successor $w'$ of $u'$ in $\pi'$ and a non-empty set $W$ of vertices of $T$ such that for each valuation $\bld{\nu}'$ of all parameters in $w'.s$ such that $w'.\varphi\prm{\bld{\nu}'}$ is satisfiable, there is $w \in W$ such that $w.s \equiv w'.s\prm{\bld{\nu}'}$. And we further show that each $w \in W$ is a successor of some $u \in U$ and there is no vertex $v'$ between $u'$ and $w'$ in $\pi'$ such that successors of $v'.s$ are computed at line~\ref{alg:executeSymbolically:SuccessorInCalling}. There are four possible cases in Algorithm~\ref{alg:executeSymbolically} for $u'.s$:

    (1) We reach line~\ref{alg:executeSymbolically:Unrolling}: Let $w'$ be a direct successor of $u'$ in $\pi'$. Obviously, $w'.s$ is one of the states $s'$ computed at line~\ref{alg:executeSymbolically:SuccessorInUnrolling}. Let $i$ be the index, for which $w'.s = s'$. The formula $w'.\varphi$ is satisfiable, since $w'$ is in $T'$ (see condition at line~\ref{alg:executeSymbolically:AddSATSuccessors}). Let $\bld{\nu}$ be a valuation for which $w'.\varphi$ is satisfiable. And let $\bld{\nu}' = \bld{\nu} \smallsetminus \{ (\kappa,\nu) \}$, where $\nu$ is an integer assigned in $\bld{\nu}$ to the fresh parameter $\kappa$ introduced at line~\ref{alg:executeSymbolically:FreshKappa}. From line~\ref{alg:executeSymbolically:SuccessorInUnrolling} we see that $u'.\varphi\prm{\bld{\nu}'}$ is satisfiable. Therefore, there is a vertex $u \in U$ such that $u.s \equiv u'.s\prm{\bld{\nu}'}$. According to Definition~\ref{def:Templates}~(L3) there is a successor $w$ of $u$ in $T$ such that
    \begin{align*}
        w.s &\equiv u.s \circ (t.\theta_i\prm{\kappa}, t.\varphi_i\prm{\kappa}, t.\Xi_i\prm{\kappa}, t.l_i)\prm{\{(\kappa,\nu)\}}
        \\
        &\equiv u'.s\prm{\bld{\nu}'} \circ (t.\theta_i\prm{\kappa}, t.\varphi_i\prm{\kappa}, t.\Xi_i\prm{\kappa}, t.l_i)\prm{\{(\kappa,\nu)\}}
        \\
        &\equiv (u'.s \circ (t.\theta_i\prm{\kappa}, t.\varphi_i\prm{\kappa}, t.\Xi_i\prm{\kappa}, t.l_i))\prm{\bld{\nu}}
        \\
        &\equiv w'.s\prm{\bld{\nu}}.
    \end{align*}
    Therefore, $w \in W$.

    (2) We reach line~\ref{alg:executeSymbolically:Calling}: Let $w'$ be a direct successor of $u'$ in $\pi'$. Obviously, $w'.s$ is one of the states $s'$ computed at line~\ref{alg:executeSymbolically:SuccessorInCalling}. Let $i$ be the index, for which $w'.s = s'$. The formula $w'.\varphi$ is satisfiable, since $w'$ is in $T'$ (see condition at line~\ref{alg:executeSymbolically:AddSATSuccessors}). Let $\bld{\nu}$ be a valuation for which $w'.\varphi$ is satisfiable. And let $\bld{\nu}' = \bld{\nu} \smallsetminus \{ (\kappa,\nu) \}$, where $\nu$ is an integer assigned in $\bld{\nu}$ to the fresh parameter $\kappa$ introduced at line~\ref{alg:executeSymbolically:FreshKappa}. From line~\ref{alg:executeSymbolically:SuccessorInCalling} we see that $u'.\varphi\prm{\bld{\nu}'}$ is satisfiable. Therefore, there is a vertex $u \in U$ such that $u.s \equiv u'.s\prm{\bld{\nu}'}$. According to Definition~\ref{def:Templates}~(R3) there is a successor $w$ of $u$ in $T$ such that
    \begin{align*}
        w.s &\equiv u.s \circ (t.\theta_i\prm{\kappa}, t.\varphi_i\prm{\kappa}, [(t,\kappa)] \circ t.\Xi_i\prm{\kappa}, t.l_i)\prm{\{(\kappa,\nu)\}}
        \\
        &\equiv u'.s\prm{\bld{\nu}'} \circ (t.\theta_i\prm{\kappa}, t.\varphi_i\prm{\kappa}, [(t,\kappa)] \circ t.\Xi_i\prm{\kappa}, t.l_i)\prm{\{(\kappa,\nu)\}}
        \\
        &\equiv (u'.s \circ (t.\theta_i\prm{\kappa}, t.\varphi_i\prm{\kappa}, [(t,\kappa)] \circ t.\Xi_i\prm{\kappa}, t.l_i))\prm{\bld{\nu}}
        \\
        &\equiv w'.s\prm{\bld{\nu}}.
    \end{align*}
    Therefore, $w \in W$.

    (3) We reach line~\ref{alg:executeSymbolically:Unwinding}: Let $x'$ be a predecessor of $u'$ in $\pi'$, which pushed (at line~\ref{alg:executeSymbolically:Push}) the record being at the top of $u'.\Xi$. Obviously, successors of $x'.s$ are computed at line~\ref{alg:executeSymbolically:SuccessorInCalling}. Further, let $y'$ and $v'$ be direct successors of $x'$ and $u'$ in $\pi'$ respectively. The formula $v'.\varphi$ is satisfiable, since $v'$ is in $T'$ (see condition at line~\ref{alg:executeSymbolically:AddSATSuccessors}). Note that $v'$ is the only successor of $u'$ in $T'$. Let $\bld{\nu}$ be a valuation for which $v'.\varphi$ is satisfiable. Note that $\bld{\nu}$ defines an integer $\nu = \bld{\nu}(\kappa)$ for the parameter $\kappa$ retrieved from stack $u'.\Xi$ at line~\ref{alg:executeSymbolically:UnwindingGetKappa}, since $y'.s$ must have already used it. From line~\ref{alg:executeSymbolically:SuccessorInUnwinding} we see that $u'.\varphi\prm{\bld{\nu}}$ is satisfiable. Therefore, there is a vertex $u \in U$ such that $u.s \equiv u'.s\prm{\bld{\nu}}$. Let $\pi$ be a path in $T$ from the root to a leaf vertex and going through $u$. According to connections between vertices of sets $U$ constructed for vertices $u'$ along $\pi'$, there is a predecessor $x$ of $u$ in $\pi$, such that $x.s \equiv x'.s\prm{\bld{\nu}}$. Since $y'$ is the direct successor of $x$ in $\pi$ (i.e.~there was computed a set $W$ for $y'$), there must also exist a vertex $y \in \pi$ lying between $x$ and $u$ and $y.s \equiv y'.s\prm{\bld{\nu}}$. Note that $u$ is the first successor of $y$ in $\pi$ with $u.l$ being an exit location and $u.\Xi = y.\Xi$. Otherwise we would apply this case (3) for some other vertex lying between $y'$ and $u'$ in $\pi'$. Therefore, from Definition~\ref{def:Templates}~(R2) there is a non-leaf vertex $v$ in a suffix of $\pi$ starting with $u$ such that
    \begin{align*}
        v.s &\overset{g}{\equiv} (u.\theta \circ t.\theta\prm{\kappa}, u.\varphi, x.\Xi, t.l')\prm{\{(\kappa,\nu)\}}
        \\
        &\overset{g}{\equiv} (u'.\theta\prm{\bld{\nu}} \circ t.\theta\prm{\kappa}, u'.\varphi\prm{\bld{\nu}}, \texttt{pop(}u'.\Xi\texttt{)}\prm{\bld{\nu}}, t.l')\prm{\{(\kappa,\nu)\}}
        \\
        &\overset{g}{\equiv} (u'.\theta \circ t.\theta\prm{\kappa}, u'.\varphi, \texttt{pop(}u'.\Xi\texttt{)}, t.l')\prm{\bld{\nu}}
        \\
        &\overset{g}{\equiv} v'.s\prm{\bld{\nu}}.
    \end{align*}
    And finally Lemma~\ref{lem:GEquivToEquiv} ensures there are the only direct successors $w$ and $w'$ of $v$ and $v'$ respectively, such that $w.s \equiv w'.s\prm{\bld{\nu}}$. Therefore, $w \in W$.

    (4) Otherwise, we reach line~\ref{alg:executeSymbolically:ClassicStep}: Let $u$ be any vertex in $U$. Since $u.s \equiv u'.s\prm{\bld{\nu}}$ for some valuation $\bld{\nu}$ for which $u'.\varphi\prm{\bld{\nu}}$ is satisfiable and since all direct successors of both $u$ and $u'$ are computed by classic symbolic execution step, there must be a direct successor $w$ of $u$ in $T$ and a direct successor $w'$ of $u'$ in $T'$ such that $w.s \equiv w'.s\prm{\bld{\nu}}$. Note that both $u'.s$ and $w'.s$ have exactly the same parameters. Therefore, $w \in W$.
\end{proof}

\section{Computation of Templates} \label{sec:Detection}

In this section we show one possible approach to computation of templates with one parameter. We provide detailed description of an algorithm computing a template for a program part with specified cyclic path, entry location, and several exit ones. Then we extend concept of the algorithm to computation of recursion templates for program parts.


\subsection{Template for Program Part with Cyclic Path} \label{sec:Detection:Cycles}

Let $P$ be a program and let us suppose we have a program part of $P$ with a cyclic path, an entry location $e$ and some exit location $x$ (but there can be other exits from the part). We show how to compute a symbolic memory $\theta_x\prm{\kappa}$, a path condition $\varphi_x\prm{\kappa}$ and a call stack $\Xi_x\prm{\kappa}$ at the exit location $x$. The computation of remaining parts of resulting template are then straightforward.

The algorithm proceeds in two steps. First, we compute a program state $(\theta, \varphi, [], e)$ resulting from classic symbolic execution of the cyclic path of the part exactly once, and a program state $(\hat{\theta}, \hat{\varphi}, \hat{\Xi}, x)$ resulting from classic symbolic execution of a path from $e$ to $x$. The second step is to express $\theta_x\prm{\kappa}$, $\varphi_x\prm{\kappa}$ and $\Xi_x\prm{\kappa}$ in terms of the program states computed in the first step.

The computation of program states $(\theta, \varphi, [], e)$ and $(\hat{\theta}, \hat{\varphi}, \hat{\Xi}, x)$ requites to run classic symbolic execution on the analysed program part. But Algorithm~\ref{alg:executeSymbolically} can only execute programs satisfying Definition~\ref{def:Program}. Therefore, we create a new program, say $P'$, representing the analysed part.

We start with a program $P'$ consisting of all variables of $P$ and of all those functions of $P$ having at least one location of the cycle. Note that the cyclic path of the part may traverse several functions through call sites. We now remove all the locations and edges in $P'$, which do not belong to the cycle nor to the path from $e$ to $x$. We assume that $x$ does not belong to the cyclic path, since otherwise we can always create its copy outside the cycle. Next we mark the function in $P'$ containing the entry location $e$ as the starting function of $P'$ and we set $e$ to be the entry location of the function. Then we create a new location $e'$ representing the exit location from the starting function. Now we break the cyclic path in the entry $e$ such that we redirect the only in-edge of $e$ (belonging to the cycle) to $e'$. And finally we transform $x$ to error location by adding loop edge with \texttt{skip} action.

$P'$ is now a program according to Definition~\ref{def:Program}. So, we can run unmarked version of Algorithm~\ref{alg:executeSymbolically}. Note that the algorithm must always terminate for $P'$. Let $E$ be a set of resulting program states. Then $|E| \leq 2$. If there is no $s \in E$ such that $s.l = e$, then we do not create the template for the part, since there is no real path around the cycle. If there is no state $s \in E$ such that $s.l = x$, then we discard the exit $x$ from the consideration for the template, since it is impossible to leave the loop through $x$. Otherwise, $E$ contains exactly two program states, which are the states we are looking for.

Now we show how to express $\theta_x\prm{\kappa}$, $\varphi_x\prm{\kappa}$ and $\Xi_x\prm{\kappa}$ in terms of the program states computed above. Let $T$ be a symbolic execution tree of $P$, computed by $*$-version of Algorithm~\ref{alg:executeSymbolically}. Further, let $u$ be a vertex of $T$ such that $u.l = e$ and $\pi = u \ldots u_1 \ldots u_2 \ldots u_{\nu} \ldots w$ be a path in $T$ starting at $u$, iterating the cycle of the part exactly $\nu \geq 0$ times, i.e.~all the vertices $u_i$ have $u_i.l = e$, and then $\pi$ leaves the cycle into the vertex $w$ with $w.l = x$. We use memory composition to express memories of vertices along $\pi$ as follows.
\begin{align*}
    u_1.\theta &= u.\theta \circ \theta
    \\
    u_2.\theta &= u_1.\theta \circ \theta = u.\theta \circ (\theta \circ \theta)
    \\
    \cdots
    \\
    u_{\nu}.\theta &= u_{\nu-1}.\theta \circ \theta = u.\theta \circ (\underset{\nu}{\underbrace{\theta \circ \cdots \circ \theta}}).
\end{align*}
If we denote the composition of $i$ symbolic memories $\theta$ by $\theta^i$, where $\theta^0 = \Theta$ and $\theta^1 = \theta$, then we have $u_{i}.\theta = u.\theta \circ \theta^i$ and we get
\begin{align*}
    w.\theta = u.\theta \circ (\theta^{\nu} \circ \hat{\theta}).
\end{align*}
We proceed similarly to express path conditions of vertices along $\pi$.
\begin{align*}
    u_1.\varphi &\equiv u.\varphi \wedge u.\theta\ese{\varphi} \equiv u.\varphi \wedge (u.\theta \circ \theta^0)\ese{\varphi} \equiv u.\varphi \wedge u.\theta\ese{\theta^0\ese{\varphi}}
    \\
    u_2.\varphi &\equiv u_1.\varphi \wedge u_1.\theta\ese{\varphi} \equiv u.\varphi \wedge u.\theta\ese{\theta^0\ese{\varphi}} \wedge (u.\theta \circ \theta^1)\ese{\varphi} \equiv u.\varphi \wedge u.\theta\ese{\theta^0\ese{\varphi} \wedge \theta^1\ese{\varphi}}
    \\
    \cdots
    \\
    u_{\nu}.\varphi &\equiv u_{\nu-1}.\varphi \wedge u_{\nu-1}.\theta\ese{\varphi} \equiv u.\varphi \wedge u.\theta\ese{\underset{\nu}{\underbrace{\theta^0\ese{\varphi} \wedge \ldots \wedge \theta^{\nu-1}\ese{\varphi}}}}
\end{align*}
Using the following equivalence
    \[\theta^0\ese{\varphi} \wedge \ldots \wedge \theta^{\nu-1}\ese{\varphi} \equiv 0 \leq \nu \wedge \forall \tau~(0 \leq \tau < \nu \rightarrow \theta^{\tau}\ese{\varphi}),\]
we can write
\begin{align*}
    w.\varphi &\equiv u_{\nu}.\varphi \wedge u_{\nu}.\theta\ese{\hat{\varphi}} \equiv u.\varphi \wedge u.\theta\ese{\theta^0\ese{\varphi} \wedge \ldots \wedge \theta^{\nu-1}\ese{\varphi} \wedge \theta^{\nu}\ese{\hat{\varphi}}}
    \\
    &\equiv u.\varphi \wedge u.\theta\ese{0 \leq \nu \wedge \forall \tau~(0 \leq \tau < \nu \rightarrow \theta^{\tau}\ese{\varphi}) \wedge \theta^{\nu}\ese{\hat{\varphi}}}.
\end{align*}
SMT solvers do not support memory composition operation appearing in the formula $w.\varphi$. Therefore, we need an equivalent \emph{declarative} description of the operation. Such a description is a parametrised symbolic memory $\theta\prm{\kappa}$, where we require $\theta\prm{\kappa} \equiv \theta^{\kappa}$, for any $\kappa \geq 0$. For a given symbolic memory $\theta$ we compute content of $\theta\prm{\kappa}$ per variable by applying the following two rules
\begin{gather*}
    \frac{\theta(\var{a}) = \Theta(\var{a}) + c,~~~\var{a}~\textrm{is of a numeric type},~~~c \textrm{~is a numeric constant of \var{a}'s type}}{\theta\prm{\kappa}(\var{a}) = \Theta(\var{a}) + c \cdot \texttt{typeOf<a>(}\kappa\texttt{)}},
    \\ \\
    \frac{\theta(\var{A}) = \Theta(\var{A}),~~~\var{A}~\textrm{is of a an array type}}{\theta\prm{\kappa}(\var{A}) = \Theta(\var{A})},
\end{gather*}
where expression $\texttt{typeOf<a>(}\kappa\texttt{)}$ represent casting operation of $\kappa$ to a type of variable \var{a}. If there is a variable, which does not match any of the rules, then we fail to compute $\theta\prm{\kappa}$. And we thus fail to compute the template. Obviously, one can provide more rules for more complex symbolic memories. The presented rules are only supposed to illustrate the process.

Having $\theta\prm{\kappa}$ we define
\begin{align*}
    \theta_x\prm{\kappa} &= \theta\prm{\kappa} \circ \hat{\theta}
    \\
    \varphi_x\prm{\kappa} &= 0 \leq \kappa \wedge \forall \tau~(0 \leq \tau < \kappa \rightarrow \theta\prm{\tau}\ese{\varphi}) \wedge \theta\prm{\kappa}\ese{\hat{\varphi}}
    \\
    \Xi_x\prm{\kappa} &= \theta\prm{\kappa} \circ \hat{\Xi},
\end{align*}
and we get $w.\theta \equiv u.\theta \circ \theta_x\prm{\nu}$, $w.\varphi \equiv u.\varphi \wedge u.\theta\ese{\varphi_x\prm{\nu}}$ and $w.\Xi \equiv u.\Xi \circ (u.\theta \circ \Xi_x\prm{\nu})$. Using these equivalences we write $w.s \equiv u.s \circ (\theta_x, \varphi_x, \Xi_x, x)\prm{\nu}$, which is exactly the equivalence used in Definition~\ref{def:Templates}~(L2) and (L3).

\subsection{Template for Program Parts Representing Recursion}

Let $P$ be a program, $f$ be a recursive function of $P$, $e$ and $x$ be entry and exit locations of $f$ respectively and let $h = (u,v)$ be an edge of $f$ with an action representing recursive call of $f$. We transform computation of recursion template for recursive calling of $f$ into analysis of two program parts $P_1$ and $P_2$ with cyclic paths. The cycle of $P_1$ starts at location $e$ and leads to $u$. We then enclose the cycle by an artificial edge whose action simulate an effect of any call of $f$. Let $e$ be entry location of $P_1$ and let $x_1, \ldots, x_n$ be its exit locations. We compute a template $t_1 = (e, n, \{ (\theta_1, \varphi_1, \Xi_1, x_1)\prm{\kappa}, \ldots, (\theta_n, \varphi_n, \Xi_n, x_n)\prm{\kappa} \})$ for $P_1$ according to algorithm from Section~\ref{sec:Detection:Cycles}. Having $t_1$ we can express the resulting recursive template $t$ as follows.
\[ t = (e, n, \{ (\theta_1, \varphi_1, \Xi_1, x_1)\prm{\kappa}, \ldots, (\theta_n, \varphi_n, \Xi_n, x_n)\prm{\kappa} \}, \theta\prm{\kappa}, x), \]
where $\theta\prm{\kappa}$ is the only unknown component in $t$. We compute the symbolic memory $\theta$ from analysis of the second program part $P_2$. The cycle of $P_2$ starts at $x$. There we add an artificial edge, whose action simulate an effect of return from any call of $f$. The artificial edge gets us to location $v$. Then we enclose the cycle by following a path from $v$ to $x$. We set $x$ to be the entry location of $P_2$ and we further set $x$ to also be the only exit location from $P_2$. As you can see, here we have introduced an assumption that there is no branching along the path from $v$ to $x$, i.e.~we cannot escape from the path. We discuss the case, when there is some branching (escape edges) along the path later. Since we have defined the program part $P_2$, we compute its template $t_2 = (x, 1, \{ (\theta\prm{\kappa}, \true, [], x) \})$ according to algorithm from Section~\ref{sec:Detection:Cycles}. Then we take the symbolic memory $\theta\prm{\kappa}$ and we complete the recursion template $t$.

\begin{figure}[!htb]
    \begin{tabular}{cc}
         {\centering
\tikzstyle{start} = [regular polygon,regular polygon sides=3,
    regular polygon rotate=180,thick,draw,inner sep=2.0pt]
\tikzstyle{target} = [regular polygon,regular polygon sides=3,
    regular polygon rotate=0,thick,draw,inner sep=1pt]
\tikzstyle{loc} = [circle,thick,draw]
\tikzstyle{pre} = [<-,shorten <=1pt,>=stealth',semithick]
\tikzstyle{post} = [->,shorten <=1pt,>=stealth',semithick]
\footnotesize
\begin{tikzpicture}[node distance=1.1cm]
    \node [start] (a) {$a$};
    \node [loc] (b) [left of=a,yshift=-10mm] {$b$}
        edge [pre,bend left] node [label=left:\texttt{i<n}] {} (a);
    \node [loc] (c) [below of=b,yshift=-5mm] {$c$}
        edge [pre] node
            [label=right:
                \begin{tabular}{l}
                    \texttt{t:=countIf(} \\
                    \texttt{~~A,i+1,n,x)}
                \end{tabular}
                ] {} (b);
    \node [loc] (d) [left of=c,yshift=-10mm] {$d$}
        edge [pre,bend left]
                node [label=below:\texttt{~~~~~~A[i]!=x}] {} (c);
    \node [loc] (e) [right of=c,yshift=-10mm] {$e$}
        edge [pre,bend right]
                node [label=above:\texttt{~~~~A[i]=x}] {} (c);
    \node [loc] (g) [right of=a,yshift=-10mm] {$g$}
        edge [pre,bend right] node [label=right:\texttt{i>=n}] {} (a);
    \node [target] (h) [below of=e] {$h$}
        edge [pre,bend left=25] node [label=left:\texttt{ret t~~}] {} (d)
        edge [pre] node [label=left:\texttt{ret t+1}] {} (e)
        edge [pre,bend right=35] node [label=left:\texttt{ret 0}] {} (g);
\end{tikzpicture}} & {\centering
\tikzstyle{start} = [regular polygon,regular polygon sides=3,
    regular polygon rotate=180,thick,draw,inner sep=2.0pt]
\tikzstyle{target} = [regular polygon,regular polygon sides=3,
    regular polygon rotate=0,thick,draw,inner sep=1pt]
\tikzstyle{loc} = [circle,thick,draw]
\tikzstyle{pre} = [<-,shorten <=1pt,>=stealth',semithick]
\tikzstyle{post} = [->,shorten <=1pt,>=stealth',semithick]
\footnotesize
\begin{tikzpicture}[node distance=1.1cm]
    \node [start] (a) {$a$};
    \node [loc] (b) [left of=a,yshift=-10mm] {$b$}
        edge [pre,bend left] node [label=left:\texttt{i<n}] {} (a);
    \node [loc] (d) [left of=b,yshift=-10mm] {$d$}
        edge [pre,bend left]
                node [label=below:\texttt{~~~~~~A[i]!=x}] {} (b);
    \node [loc] (c) [below of=d,yshift=-5mm] {$c$}
        edge [pre,inner sep=0.5pt] node
            [label=right:
                \begin{tabular}{l}
                    \texttt{t:=countIf(} \\
                    \texttt{~~A,i+1,n,x)}
                \end{tabular}
                ] {} (d);
    \node [loc] (e) [right of=b,yshift=-10mm] {$e$}
        edge [pre,bend right]
                node [label=above:\texttt{~~~~A[i]=x}] {} (b);
    \node [loc,inner sep=1.5pt] (c2) [below of=e,yshift=-5mm] {$c'$}
        edge [pre,inner sep=0.5pt] node
            [label=right:
                \begin{tabular}{l}
                    \texttt{t:=countIf(} \\
                    \texttt{~~A,i+1,n,x)}
                \end{tabular}
                ] {} (e);
    \node [loc] (g) [right of=a,xshift=10mm,yshift=-15mm] {$g$}
        edge [pre,bend right] node [label=right:\texttt{i>=n}] {} (a);
    \node [target] (h) [below of=c2] {$h$}
        edge [pre,bend left=25] node [label=left:\texttt{ret t~~}] {} (c)
        edge [pre] node [label=left:\texttt{ret t+1}] {} (c2)
        edge [pre,bend right=55] node [label=left:\texttt{ret 0}] {} (g);
\end{tikzpicture}}
         \\
        (a) & (b)
    \end{tabular}
    \caption{Two equivalent recursive implementations of the function \texttt{countIf(A,i,n,x)}.}
    \label{fig:countIfRec}
\end{figure}
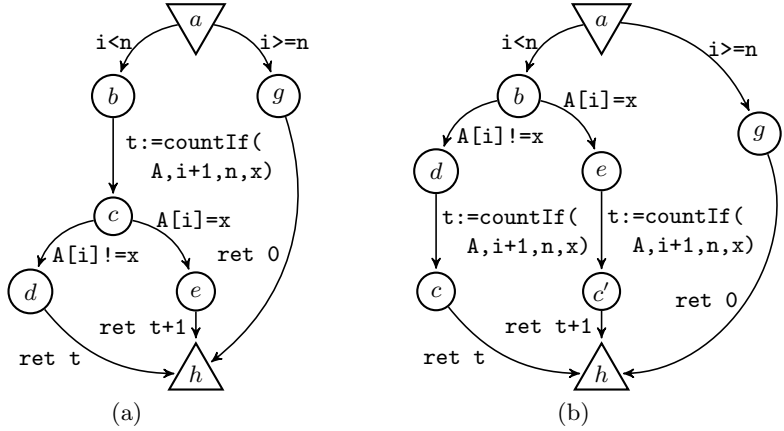

Note that we can simplify computation of $\theta\prm{\kappa}$ of the template $t_2$ such that we only express a return value of $f$. We do not need to express local variables of $f$, since requirement (R2) of Definition~\ref{def:Templates} uses the equivalence $\overset{g}{\equiv}$. We further note, that the algorithm above also works for indirect recursion. It immediately follows from the algorithm in Section~\ref{sec:Detection:Cycles}, where cyclic path of an analysed program part may traverse several functions.

We finish the section by discussion of the assumption we gave to the cyclic path of $P_2$. We assumed there is no branching along the path from $v$ to $x$. The algorithm presented above can compute templates for tail recursions and for many non-tail ones, while keeping the computation simple (we only need $\theta\prm{\kappa}$ expressed just for return value). Therefore, we believe the assumption has only small impact to applicability of the algorithm. Besides, it is always possible to move edges with recursive calls below branchings not depending on return values form the calls. We demonstrate this process at Figures~\ref{fig:countIfRec}~(a) and (b), where we depict two equivalent recursive implementations of the function \texttt{countIf}. We can easily check that in program at Figures~\ref{fig:countIfRec}~(b) there are two program parts (one per recursive call), for which we can compute templates according to the algorithm described above.

\section{Discussion}

We presented compact symbolic execution using only templates with a single parameter. We further restrict ourselves to computation of templates only for program parts consisting of cyclic paths of representing recursion. We can get even better reduction of size of symbolic execution tree, if we create templates for more complex program parts, and when we use more parameters. Let us consider the function \texttt{countIf} at Figure~\ref{fig:countIf}. The program loop in the function consists of two cyclic paths around it. We have already discussed templates for both cycles in Section~\ref{sec:Overview}. But if we built a single template using two parameters (one parameter per cyclic path), then resulting compact symbolic execution tree would be finite. We see, there is a space for extensions of the basic concepts we presented here. 

Let us consider well known algorithm \texttt{binarySearch}. Template detection for this program (even with a single parameter) may infer geometric progressions as values of some variables. They may later cause serious performance issues for SMT solver, when they get into a path condition.

Compact symbolic execution commonly has higher performance requirements to SMT solvers then classic one. Path conditions may contain template parameters besides symbols. And parameters are quantified. This is the price of the ability to reason about multiple program paths at once.

King showed effectiveness of symbolic execution for automated testing generation~\cite{Kin76}. Producing a good test typically means to reach some interesting (e.g.~bug suspicious) program location. Compact symbolic execution can be very helpful in this task. Let us consider a situation, when reachability of such a target location is dependant on an exact number of iterations of a particular cycle. Providing a template for a program part with the cycle, we can simultaneously reason about all the paths exiting from the cycle. Therefore, instead of exploration of paths space by classic symbolic execution, we can just send a query to SMT solver to check satisfiability of parametrised path condition.

King also showed in his paper~\cite{Kin76}, how symbolic execution can be used in proving program correctness according to Floyd's method~\cite{Floyd67}. Using templates we can decrease or in some cases even eliminate the need of loop invariants. For programs, where compact symbolic execution is finite in contrast to classic one, there we do not need loop invariants at all. And for other programs, loop templates describe behaviour of some paths through loop, and we may therefore provide simpler invariants for the remaining behaviour of the loop.

\section{Related Work}

Compact symbolic execution is tightly related to the work of King in 1976~\cite{Kin76}, where the author introduced the general concept of classic symbolic execution. Besides the description of symbolic execution King discussed its applicability to program testing and formal proving of correctness according to Floyd's method~\cite{Floyd67}. Nevertheless, issues like the path explosion problem were not tackled.

In~\cite{KPV03} authors propose a program instrumentation by a code providing lazy initialisation of dynamically allocated data structures like lists or trees and they enable symbolic execution of the instrumented program by a standard model checker without building a dedicated tool. The lazy initialisation algorithm is further improved and formally defined as an operational semantics of a core subset of the Java Virtual Machine in~\cite{DLR11}.

A scalability of symbolic execution to real world programs can be improved by exploring only client's code~\cite{KS05}. A library code (like string manipulation, standard containers like sets or maps) can be assumed as well defined and properly tested.

There are several symbolic execution based techniques constructing loop summaries or simply counting loop iterations~\cite{GL11,SPmCS09,ST11}. The introduction of counters usually provides a possibility to speak about multiple paths through loop at once.
A technique presented in~\cite{GL11} analyses loops on-the-fly, i.e.~during simultaneous concrete and symbolic execution of a program for a concrete input. The loop analysis infers inductive variables, i.e. variables that are modified by a constant value in each loop iteration. These variables are used to build loop summaries expressed in a form of pre and postconditions. 
The LESE technique presented in~\cite{SPmCS09} introduces symbolic variables for the number of times each loop was executed. LESE links the symbolic variables with features of a known grammar generating inputs. Using these links, the grammar can control the numbers of loop iterations performed on a generated input.
A symbolic-execution-based algorithm in~\cite{ST11} produces a nontrivial necessary condition on input values to drive the program execution to the given location. The key part of the technique is computation of loop summaries in form of symbolic program states and path conditions both parametrised by so called path counters. Each path counter is assigned to individual path through the analysed loop.

There are also approaches computing function summaries~\cite{G07,AGT08}. Reusing summaries at call sites typically leads to an interesting performance improvement. Moreover, summaries may insert additional symbolic values into a path condition which often leads to another performance improvement.

Finally, there are also techniques partitioning program paths into separate classes according to impact of the paths to a given set of program variables~\cite{QNR11,SH10}. Values of output variables are typically considered as a partitioning criteria.


\section{Conclusion}

We introduced a generalisation of classic symbolic execution called compact symbolic execution. We generalised notion of symbols of classic symbolic execution such that symbols can be related to different program locations now. This allows us to analyse individual parts of a given program separately from the rest of the program.
We further introduced concept of templates representing declarative parametric descriptions of behaviour of separately analysed program parts. We gave precise definition of templates with one parameter and we provided algorithm of compact symbolic execution using these templates.

\bibliographystyle{plain}
\bibliography{cse} 

\end{document}